\documentclass[10pt,final,journal]{IEEEtran}

\usepackage{cite}
\usepackage{url}
\usepackage{upgreek}
\usepackage{amsmath}
\usepackage{amssymb}
\usepackage{dsfont}
\usepackage{bm}
\usepackage{enumitem}
\usepackage[caption=false,font=footnotesize]{subfig}
\usepackage{siunitx}
\usepackage{algorithm}
\usepackage{graphicx}
\usepackage{epstopdf}
\usepackage{booktabs}
\usepackage{color}
\usepackage{diagbox}
\graphicspath{{figures/}}
\usepackage{multirow}
\usepackage{multicol}
\usepackage{pifont}
\usepackage{threeparttable}
\usepackage{stfloats}
\usepackage[table]{xcolor}
\usepackage{tablefootnote}
\usepackage{array}
\usepackage{arydshln}
\def\wt{\widetilde}

\DeclareMathOperator*{\argmin}{argmin}
\def\diag{\text{diag}}

\def\sat{\text{sat}}

\def\T{\text{T}}
\def\r{\text{ref}}
\def\and{\text{and}}
\def\NN{\text{NN}}
\def\geo{\text{geo}}
\def\gra{\text{gra}}

\usepackage{tikz}
\newcommand*{\circled}[1]{\lower.7ex\hbox{\tikz\draw (0pt, 0pt)%
		circle (.5em) node {\makebox[1em][c]{\small #1}};}}

\newcommand{\mathletter}[1]{%
	\expandafter\newcommand\csname b#1\endcsname{\mathbb #1}
	\expandafter\newcommand\csname c#1\endcsname{\mathcal #1}
	\expandafter\newcommand\csname f#1\endcsname{\mathfrak #1}
	\expandafter\newcommand\csname til#1\endcsname{\widetilde #1}
	\expandafter\newcommand\csname ha#1\endcsname{\widehat #1}
	\expandafter\newcommand\csname bf#1\endcsname{\bf #1}
}%
\def\mathletters#1{\mathlettersB #1,,}
\def\mathlettersB#1,{\ifx,#1,\else\mathletter #1\expandafter\mathlettersB\fi}
\mathletters{A,B,C,D,E,F,G,H,I,J,K,L,M,N,O,P,Q,R,S,T,U,V,W,X,Y,Z}

\def\bee{\begin{equation}}
\def\ene{\end{equation}}
\def\beq{\begin{eqnarray}}
\def\enq{\end{eqnarray}}
\def\bmatri{\begin{bmatrix}}
	\def\ematri{\end{bmatrix}}

\newtheorem{lemma}{Lemma}

\newenvironment{proof}{\begin{IEEEproof}}{\end{IEEEproof}}

\def\TP{{\mathrm T}}
\def\mD{{\mathcal D}}

\def\mR{{\mathcal R}}
\title{Standoff Tracking Using DNN-Based MPC with Implementation on FPGA}

\author{Fei~Dong, Xingchen~Li, Keyou~You,~\IEEEmembership{Senior Member,~IEEE}, Shiji~Song,~\IEEEmembership{Senior Member,~IEEE} % <-this % stops a space
	\thanks{This work was supported by the National Natural Science Foundation of China (62203027 and 62033006), the China Postdoctoral Science Foundation (BX20220369 and 2022M710310), and a grant from the Guoqiang Institute, Tsinghua University. ({\em Corresponding author: Keyou You})}% <-this % stops a space
	\thanks{F. Dong is with School of Automation Science and Electrical Engineering, Beihang University, Beijing 100191, China, and was a PhD student at Tsinghua University when this work was mainly completed. E-mail: dongf@buaa.edu.cn.}
	\thanks{X. Li, K. You, and S. Song are with Department of Automation, and Beijing National Research Center for Information Science and Technology, Tsinghua University, Beijing 100084, China. E-mail: lixc21@mails.tsinghua.edu.cn, youky@tsinghua.edu.cn, shijis@tsinghua.edu.cn.} 
}

\begin{document}
	\maketitle
	\begin{abstract}
This work studies the standoff tracking problem to drive an unmanned aerial vehicle (UAV) to slide on a desired circle over a moving target at a constant height. We propose a novel  Lyapunov guidance vector (LGV) field with tunable convergence rates for the UAV's trajectory planning and a deep neural network (DNN)-based model predictive control (MPC) scheme to track the reference trajectory. Then,  we show how to collect samples for training the DNN {\em offline} and design an integral module (IM) to refine the tracking performance of our DNN-based MPC. Moreover, the hardware-in-the-loop (HIL) simulation with an FPGA@$200$\si{MHz} demonstrates that our method is a valid alternative to embedded implementations of MPC for addressing complex systems and applications which is impossible for directly solving the MPC optimization problems.   
	\end{abstract}
		
	
\section{Introduction}	
Using unmanned aerial vehicles (UAVs) for convoy protection and aerial surveillance has attracted broad attention. To achieve it, it is essential for the UAV to  {\em standoff track} a moving target, which not only permits a persistent tracking but also expands its surveillance area. In the literature, it is also called the target circumnavigation or encirclement  \cite{kokolakis2021robust,dong2019Target,matveev2017tight,swartling2014collective}.

	The first of this strategy lies in designing a good guidance law. In \cite{Lawrence2003Lyapunov}, a Lyapunov guidance vector (LGV) is initially proposed  and then extended in \cite{Frew2008Coordinated,pothen2017curvature,Dong2019Flight}. The issue on the non-adjustable convergence rate of these works is resolved in  \cite{Kapitanyuk2017A} with a gradient-based method. In this work, we directly propose a new LGV with easily tunable convergence rates, which contains the GVs in \cite{Lawrence2003Lyapunov,Frew2008Coordinated,pothen2017curvature,Dong2019Flight} as a special case.     
	
	The second is to regulate the UAV to track the LGV in the complex environment subject to dynamics constraint and state-control vector constraints. If there is no constraint, it is conceivable  to utilize the backstepping controller by regarding the guidance law as the reference signal. For example, Ref. \cite{ghommam2016quadrotor} adopts a structure of five backstepping steps that resembles the multi-loop proportional-integral-derivative (PID) controller in \cite{luukkonen2011modelling,gati2013open,yao2021singularity}.  While they cannot explicitly handle the state and control constraints, the model predictive control (MPC) becomes a natural alternative \cite{rawlings2017model}.  Since the MPC law is usually obtained by solving online optimization problems,  it is only applicable to ``simple" systems with moderate sizes \cite{compare_favato2021integral,cimini2020embedded,rubagotti2016real,bemporad2002explicit,karg2020efficient}.  If the size of the MPC optimization problem is small for linear time-invariant (LTI) system,  the so-called explicit MPC \cite{bemporad2002explicit} appears to be perfect as its control law has a piece-wise affine (PWA) form, meaning that it can be obtained via a simple search over a finite number of polyhedra. Unfortunately, this number grows {\em exponentially} with respect to the MPC size, including the numbers of prediction horizons, decision variables and their constraints. For instance, such a number is up to $1638$ for a single-input LTI system with $4$ box-constrained states and $10$ prediction horizon \cite{karg2020efficient}.
	
	How to implement the MPC on a resource-limited embedded platform in a fast and faithful way is a long-standing problem. With the development of the deep learning and the advanced RISC machine (ARM) or field-programmable gate array (FPGA), such a problem has attracted resurgent interest. 
Roughly speaking, the related works can be categorized into two types of approaches, depending on whether the MPC law is computed iteratively or approximated via a deep neural network (DNN). The works on developing tailored iterative solvers for an embedded platform include \cite{hartley2013predictive,lucia2017optimized,quirynen2018embedded,xu2021custom}. Specifically, Ref. \cite{quirynen2018embedded}  exploits the structure of the MPC optimization problem and develops \texttt{PRESAS} for online solving the sequential quadratic programs (QP).  A particle swarm optimization (PSO) method is designed in \cite{xu2021custom} to search the MPC solution in parallel.

	The other alternative that is closely related to this work is on the use of an {\em offline} trained DNN to implement the  MPC law. For example, a DNN is trained to approximate the PWA function of the explicit MPC law in \cite{karg2020efficient}, and a primal-dual learning framework is proposed in \cite{zhang2021near} to regulate a linear parameter-varying system. In fact, the DNN-based MPC has demonstrated promising results on resonant power converters \cite{lucia2020deep}, PMSMs  \cite{abu2022deep}, and UAVs  \cite{tagliabue2022efficient} whose DNN-based MPC is however only evaluated on a high-performance NVIDIA TX2 CPU@$2$\si{GHz} with Pytorch. Moreover, the DNN is adopted to warm start a sequence of primal active-set algorithms in \cite{chen2022large}. Note that it is also unknown how many iterations are required to obtain a certifiable solution of the MPC optimization problem.  In \cite{chan2021deep}, a correction factor  is further introduced, hoping to achieve offset-free set-point tracking. Clearly, there is a tradeoff between the approximation error of the MPC law and the size of the DNN, which is key to the performance of the closed-loop system, and how to quantify their relationship remains an open problem.

	In contrast to the aforementioned works, we first adopt a DNN-based MPC for the UAV to track our LGV and then design a simple integral module (IM) to refine the tracking performance.  Since we aim at the UAV's resource-limited onboard implementation,  the DNN size  cannot be too large which inevitably induces approximation errors. Thus, we do not expect that the DNN-based MPC is able to exactly track our LGV and thus propose the IM to further reduce the steady-state tracking error, allowing to use a relatively small size of the DNN. The idea is presented in the preliminary version of this article in \cite{dong2022deep}. Table \ref{tab_cmp} indeed clarifies the computing efficiency of implementing our DNN-based MPC on an FPGA. To speedup the use of an FGPA for a control designer, we introduce a verification and implementation framework that enables  to directly generate and verify \texttt{verilog} codes from algorithms developed in MATLAB. Moreover,  the hardware-in-the-loop (HIL) simulation with an FPGA@$200$\si{MHz} demonstrates that our method is a valid alternative to embedded implementations of MPC for managing complex systems and applications. Overall, our contributions are summarized as follows.

	\begin{table*}[!t]
		\renewcommand{\arraystretch}{1.3}
		\centering
		\caption{A comparison of the existing works.}
		\begin{threeparttable}
			\begin{tabular}{cccccccccccc}
				\toprule[1pt]
				 Ref. & Method & Systems & $\mathrm{dim}(\bm x)$ & $\mathrm{dim}(\bm u)$ & $N_p$ & $N_c$ & Platform & Latency & {Latency (200\si{MHz})\tnote{1}} \\
				\midrule
				\cite{cimini2020embedded} & Active-set & {Linear} & 2 & 2 & 3 & 1 & F28335 DSP ($150$\si{MHz}) & $< 0.3$\si{ms} & $< 0.23$\si{ms} \\
				\hline
				\cite{quirynen2018embedded} & \texttt{PRESAS} & {Nonlinear} & 11 & 4 & 20 & 20 & ARM Cortex-A7 ($900$\si{MHz}) & $6.44$\si{ms} & $28.98$\si{ms} \\			
				\hline
				\multirow{2}*{\cite{xu2021custom}} & \multirow{2}*{PSO} & \multirow{2}*{Nonlinear} & \multirow{2}*{3} & \multirow{2}*{3} & \multirow{2}*{5} & \multirow{2}*{1} & ARM Cortex-A9 ($800$\si{MHz}) & $60.62$\si{ms}  & $242.48$\si{ms} \\
				~ & ~ & ~ & ~ & ~ & ~ & ~
				& Altera FPGA ($50$\si{MHz}) & $1.09$\si{ms}  & $0.27$\si{ms} \\
				\hline
				{\cite{karg2020efficient}} & {DNN} & {Linear} & 4 & 1 & 10 & 10 & {ARM Cortex-M0 ($48$\si{MHz})} & $3.8$\si{ms} & $0.912$\si{ms}\\
				\hline
				{\cite{zhang2021near}} & {DNN} & {Linear} & {4} & {3} & {3} & {3} & TC27x ECU DSP ($200$\si{MHz}) & $1.8$\si{ms} & $1.8$\si{ms} \\
				\hline
				{\cite{lucia2020deep}} & {DNN} & {Nonlinear} & {2} & {2} & {10} & {10} & {FPGA\tnote{2}}  & $2.1\upmu$\si{s} & -- \\
				\hline
				{\cite{tagliabue2022efficient}} & {DNN} & {Linear} & {8} & {3} & {30} & {30} & {NVIDIA TX2 CPU\tnote{3}~ ($2$\si{GHz})}   & $1.66$\si{ms} & $16.6$\si{ms} \\
				\hline
				\multirow{2}*{This work} & \multirow{2}*{DNN} & \multirow{2}*{Nonlinear} & \multirow{2}*{12} & \multirow{2}*{4} & \multirow{2}*{20} & \multirow{2}*{10} 
				& ARM Cortex-A53 ($200$\si{MHz}) & $1.030$\si{ms} & {$1.030$\si{ms}} \\
				~ & ~ & ~ & ~ & ~ & ~ & ~
				& Zynq FPGA ($200$\si{MHz}) & $0.126$\si{ms} & $0.126$\si{ms} \\
				\bottomrule[1pt]
			\end{tabular}
			\begin{tablenotes}
				\footnotesize
\item[1] Eq. (9) in \cite{hartley2013predictive} is adopted to covert the reported latency in the reference work to that on the same type processor @200\si{MHz}.
				\item[2] The specifics of FPGA are not provided in \cite{lucia2020deep}.
				\item[3] The NVIDIA TX2 CPU contains a dual-core NVIDIA Denver 2 processor @$2$\si{GHz} and a quad-core ARM Cortex-A57 processor @$2$\si{GHz}.
			\end{tablenotes}
		\end{threeparttable}
		\label{tab_cmp}
	\end{table*}

	\begin{itemize}
		\item [(a)] A novel LGV law with tunable convergence rates is proposed for standoff tracking a moving target, which covers the GV in \cite{Lawrence2003Lyapunov,Frew2008Coordinated,pothen2017curvature,Dong2019Flight} as a special case.    
		\item [(b)] We propose a DNN-based MPC with a simple IM to track the LGV. 
		\item [(c)] An easily accessible  framework is introduced for fast generating and verifying \texttt{verilog} codes from algorithms developed in MATLAB.
	\end{itemize}
	
The rest of the paper is organized as follows. In Section \ref{sec2}, we describe the standoff tracking problem.  We propose a new LGV for planning the standoff tracking trajectory in Section \ref{sec_3} and design MPCs to track the reference trajectory in Section \ref{sec_controller}. In Section \ref{sec5}, we collect training samples  and propose a DNN-based MPC with an IM track our LGV. The HIL simulation with an FPGA@$200$\si{MHz} is performed in Section \ref{sec6}, and some conclusion remarks are drawn in Section \ref{sec7}.
	
	{\em Notation}: Throughout this paper, scalars, vectors, and matrices are respectively denoted by lowercase letters, lowercase bold letters, and uppercase letters, e.g. $a$, $\bm a$, and $A$. The vectors $\bm 0_{n}, \bm 1_{n}$ and the matrix $I_n$ denote the zero and one vectors and an identity matrix with specified dimensions. The superscript $(\cdot)^\T$ represents the matrix transpose. The symbol $\Vert \cdot \Vert$ denotes the Euclidean norm. The function $\text{col}(\cdot)$ and $\text{blkdiag}(\cdot)$ return a column vector by stacking each vector on top of the other and a block diagonal matrix, respectively.  For a vector $\bm x$, then $\bm x \in [0,\bar x]$ means that every element of $\bm x$ belongs to this interval. Moreover, $\bm x(t)$ and $\bm x(k)$ represent the continuous- and discrete-time state vectors with the relation that $t=t_0+k\tau$, where $t_0$ is the initial time instant and $\tau$ is the sampling period. The vector ${\bm x}(i|k)$ represents  the  predicted state at time step $k+i$ based on $\bm x(k)$.

	\section{Problem Formulation} \label{sec2}
	
	The standoff tracking of this work requires a quadrotor UAV to slide on a desired circle over a moving target at a constant height, and is also called the target circumnavigation or target encirclement in \cite{kokolakis2021robust,matveev2017tight,dong2019Target,swartling2014collective}.  	

	\subsection{The kinematic model of the moving target}
	The moving target of interest travels in the single-integrator form
	\begin{align} \label{target}
	\dot {\bm p}_o(t) = \bm {v}_o(t),
	\end{align}
	where $\bm p_o(t)=[x_o(t),y_o(t),z_o(t)]^\T$ and $\bm v_o(t)\in \mathbb{R}^3$ respectively denote its position and velocity. 	
	\subsection{The dynamical model of the quadrotor UAV}
	
	To standoff track the target in \eqref{target}, we adopt a quadrotor UAV in \cite{luukkonen2011modelling} whose Newton-Euler dynamical model is expressed by  	\begin{equation}\label{model}
	\begin{split}
	\ddot {\bm \xi}(t) = & \bm g + {m}^{-1} \left(R(\bm \eta(t)) \bm t_{b}(t)-\diag(d_x, d_y, d_z)  \dot{\bm \xi}(t)\right),\\
	\ddot {\bm \eta}(t) =& J^{-1}(\bm \eta(t)) \left( \bm \tau_{b}(t) -C(\bm \eta(t), \dot {\bm \eta}(t)) \dot {\bm \eta}(t)  \right),
	\end{split}
	\end{equation}
	where ${\bm \xi}(t) = [x_q(t), y_q(t), z_q(t)]^{\T}$, $\bm \eta(t) = [\phi(t), \theta(t), \psi(t)]^{\T}$ denote the 3D position and the vector of Euler angles in the inertial frame, respectively.   $R(\bm \eta(t))$, $J(\bm \eta(t))$, $C(\bm \eta(t), \dot {\bm \eta}(t))$, $\bm t_{b}(t)$, and $\bm \tau_{b}(t)$ denote the rotation matrix from the body frame to the inertial frame, the Jacobian matrix, the Coriolis matrix,  the total thrust and torque in the body frame, respectively. Moreover,  $\bm t_{b}(t)$ and $\bm \tau_{b}(t)$  are determined by the control input to the rotor in the following form
	\begin{equation} \label{thrust}
	\begin{split}
	\bm t_{b}(t) &= \bmatri 0, & 0, & l_c \sum_{i=1}^{4} u_i(t)\ematri ^{\T}\in\mathbb{R}^3,\\
	\bm \tau_{b}(t) &= \bmatri l \cdot l_c \left(-u_2(t) +u_4(t)\right)  \\ l\cdot l_c \left(-u_1(t)+u_3(t)\right) \\ d_c\left(-u_1(t)+u_2(t)-u_3(t)+u_4(t)\right)   \ematri \in\mathbb{R}^3,
	\end{split}
	\end{equation} 
	where $\bm u(t)= [u_1(t), u_2(t), u_3(t), u_4(t)]^\T \in[0,\bar u]$ and other unspecified  quantities are constant parameters. More details of \eqref{model} can be found in \cite[Chapter 2]{luukkonen2011modelling}.  In simulations,  we  directly obtain the  codes  from \cite{mathworks} for the UAV's dynamics with a sampling period $\tau=0.1$.
	
	\subsection{The objective of standoff tracking}
	Let the relative 3D position between the target and UAV be
		\begin{equation}\label{eq_rea}
		[x(t),y(t),z(t)]^{\T}=[x_q(t) -x_o(t),y_q(t) -y_o(t),z_q(t) -z_o(t)]^{\T}
		\end{equation} 
		and define their relative velocity and distance in the horizon plane as
		\begin{equation}\label{eq_rang}
		\bm v(t) =[\dot x(t),\dot y(t)]^{\T}, ~~r(t)= \left(x^2(t) + y^2(t)\right)^{1/2}.
		\end{equation} 
	Then, the standoff tracking requires the UAV in \eqref{model} to loiter over the target in \eqref{target} such that
	\begin{align}\label{obj}
	\lim_{t\rightarrow \infty} r(t) = r_d,~ \lim_{t\rightarrow \infty} z(t) = z_d,~ \and~\lim_{t\rightarrow \infty} \Vert \bm v(t) \Vert = v_d,   
	\end{align}
where the triple $[r_d, z_d, r_d]^\T$ denotes the desired standoff tracking pattern. 
	
	To achieve \eqref{obj},  the maximum  of the UAV's speed  $\Vert \dot{\bm \xi}(t) \Vert$ is assumed to be sufficiently greater than that of the target \cite{Frew2008Coordinated}, and the UAV is able to localize the target with onboard sensors \cite{Dong2019Flight}.

	\section{Lyapunov Guidance Vector for Standoff Tracking} \label{sec_3}

	In general, the guidance commands  for the standoff tracking of the UAV are separately designed in the horizontal plane and its vertical direction. Since the guidance in the vertical direction is trivial, this section focuses on the guidance design in the horizontal plane via a new LGV. 
	\subsection{The gradient-based LGV}
	
	\begin{figure}[!t]
		\centering{\includegraphics[width=0.7\linewidth]{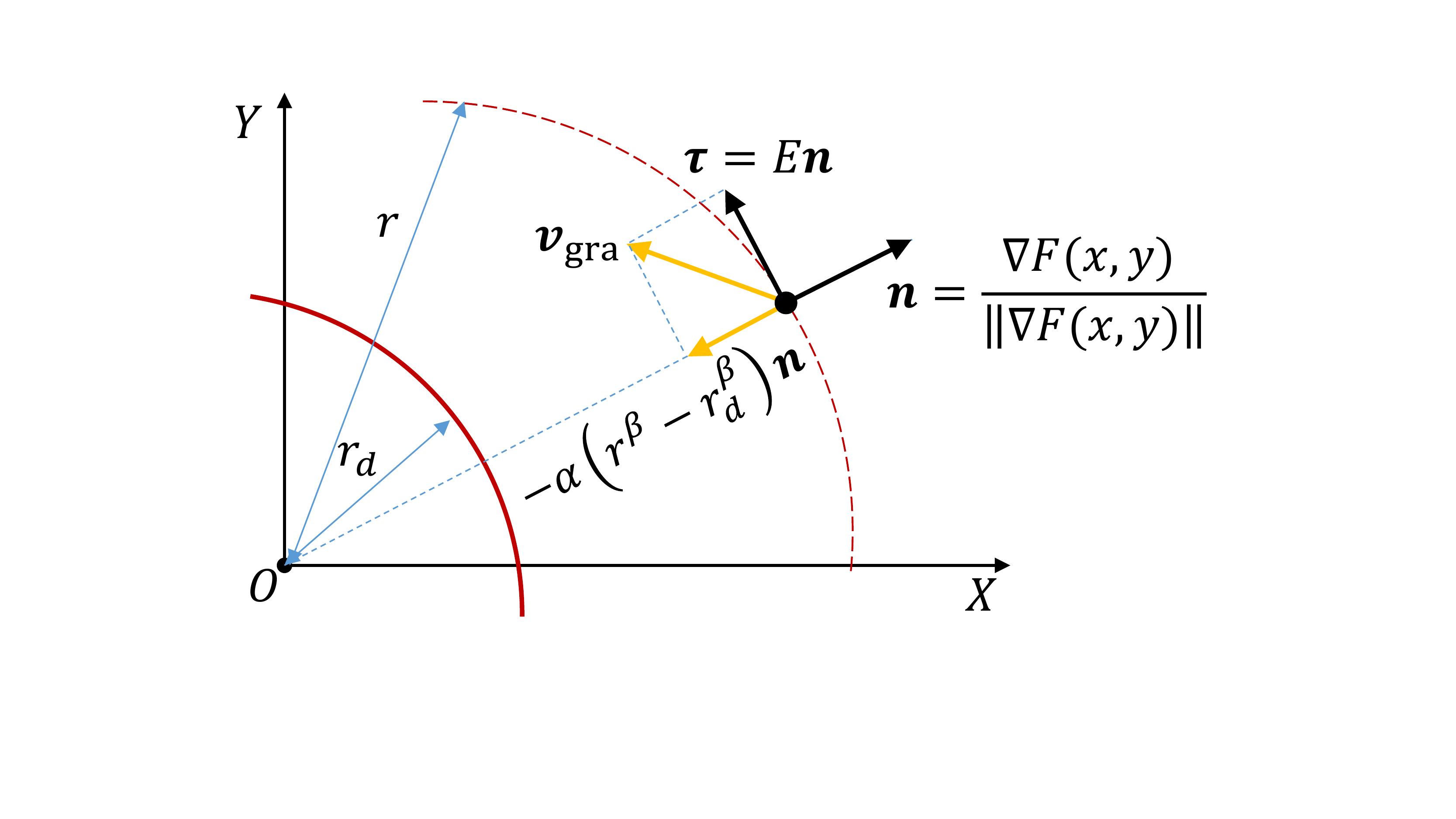}}
		\caption{Illustration of the gradient-based LGV in \eqref{eqgradient}.}
		\label{figgradient}
	\end{figure}
	We build a Cartesian frame with the target position as its origin  and define a scalar field as follows
	\begin{align}\label{eqf}
	F(x,y)=x^2+y^2.
	\end{align}
	Then, the objective of $\lim_{t \rightarrow \infty} r(t) =r_d$ requires the UAV to slide on the following isoline \cite{dong2021isoline}
	\[ \mathcal{L}(r_d)= \left\{ [x,y]^\T\in\mathbb{R}^2|F(x,y)=r_d^2 \right\},\]
	which is represented by the red solid arc in Fig.~\ref{figgradient}.   
	To this goal, a term of $- \alpha \left(r^\beta(t) -r_d^\beta\right)  {\bm n}(t)$ is naturally needed for the GV to pull the UAV towards $ \mathcal{L}(r_d)$ at the negative gradient direction with a force depending on the error of $r^\beta(t)-r_d^\beta$, where $\alpha$ is a positive parameter, $\beta \in \mathbb{N}_+$ is a constant exponent, $r(t)$ is the range (c.f. \eqref{eq_rang}) from the UAV to the target, and $\bm n(t)$ denotes the unit gradient vector, i.e.,
	\begin{align}\label{eqnab}
	{\bm n}(t)= \frac{\nabla F(x(t),y(t))}{\Vert \nabla F(x(t),y(t)) \Vert} = \bmatri x(t)/r(t)\\ y(t)/r(t) \ematri
	\end{align} 
	where $\nabla F(x,y) = \bmatri 2 x, 2y\ematri^{\TP}$ is  the gradient of \eqref{eqf}.  
	If the UAV reaches $\mathcal{L}(r_d)$, i.e., $r(t)=r_d$, it is required to move along $\mathcal{L}(r_d)$. Thus, an orthogonal vector $\bm \tau(t)$ to the unit gradient vector ${\bm n}(t)$ is further introduced, i.e.,
	\begin{align}\label{direction}
	{\bm \tau}(t) =& \bmatri \cos \gamma &-\sin \gamma \\ \sin \gamma & \cos \gamma \ematri  {\bm n}(t) = E {\bm n}(t),
	\end{align}  
	where $\gamma = \pi/2$ or $-\pi/2$ determines the rotation angle and without loss of generality, we set $\gamma = \pi/2$ in this work.

	Then,  we design a GV as follows
	\begin{align}\label{eqgradient}
	\bm v_{\gra}(t) =  {\bm \tau}(t) - \alpha \left(r^\beta(t) -r_d^\beta\right)  {\bm n}(t).
	\end{align}
	See Fig.~\ref{figgradient} for an illustration.

	For the objective of $\lim_{t\rightarrow\infty} \Vert \bm v(t) \Vert = v_d $, we scale \eqref{eqgradient} as  
	\begin{align*} 
	&\bar {\bm v}_{\gra}(t) = v_d \frac{\bm v_{\gra}(t)}{\Vert \bm v_{\gra}(t) \Vert }\nonumber\\
	&= \frac{-v_d/r(t)}{ \sqrt{1+\alpha^2 \left(r^\beta(t)-r_d^\beta\right)^2}} \bmatri  \alpha x(t)\left(r^\beta(t)-r_d^\beta\right)  -y(t) \\ 
	\alpha y(t) \left(r^\beta(t)-r_d^\beta\right)  +x(t)\ematri.
	\end{align*}	
	
	Instead of using a constant $\alpha$ in \cite{yao2021singularity},  we further use its time-varying version to accelerate the GV convergence, i.e., 
	\begin{align} \label{eqalpha}
	\alpha(t) =\left({4 r^\beta(t) r_d^\beta}\right)^{-1/2}.
	\end{align}
	Inserting \eqref{eqalpha} into $\bar {\bm v}_{\gra}(t)$ yields our LGV, i.e., 
	\begin{align} \label{lya_cont}
	\begin{split}
	\bm v_L(t) = &~\frac{-v_d}{r(t)\left(r^\beta(t)+r_d^\beta\right) } \\
	& \times \bmatri x(t)\left(r^\beta(t) -r_d^\beta\right) + 2y(t)\sqrt{r^\beta(t) r_d^\beta} \\ y(t)\left(r^\beta(t)-r_d^\beta\right) -2x(t)\sqrt{r^\beta(t) r_d^\beta} \ematri.
	\end{split}
	\end{align}

	\begin{lemma} \label{lemma1}
		Consider a stationary target located at $[x_o,y_o]^{\T}$ and the position $[x_q(t), y_q(t)]^\T$ governed by the LGV in \eqref{lya_cont}, i.e., $[\dot x_q(t), \dot y_q(t)]^\T = \bm v_L(t)$.  Then, it holds that
		\begin{align*}
		\lim_{t\rightarrow \infty} r^\beta(t) = r_d^\beta ~\text{and}~\lim_{t\rightarrow \infty} v(t) = v_d
		\end{align*}
	\end{lemma} 
	where $v(t)=\Vert \bm v(t)\Vert$, and $r(t), \bm v(t)$ are given in \eqref{eq_rang}.
	\begin{proof}
		First,  $\lim_{t\rightarrow \infty} v(t) = v_d$ trivially holds as it follows from \eqref{lya_cont} that $\Vert\bm v_L(t)\Vert=v_d$. Then, we consider a Lyapunov function candidate as 
		\begin{align}\label{lyfun}
		L_1(r(t)) = \frac{1}{2}\left( r^\beta(t) - r_d^\beta \right)^2,
		\end{align}
		and take its derivative along with \eqref{lya_cont} 
		\begin{align} \label{eql1}
		\dot L_1(r(t)) 
		&=\beta r^{\beta-1}(t)\left(r^\beta(t) -r_d^\beta\right) {r^{-1}(t)} [x(t), y(t)]^\T \bm v_L(t)\nonumber\\
		&=-\frac{2\beta v_d r^{\beta-1}(t) }{r^\beta(t)+r_d^\beta }L_1(r(t))\\
		&\le 0.\nonumber
		\end{align}

		Clearly, $L_1(r_d)=0$, and $L_1(r(t))> 0$ for any $r(t)\neq r_d$. By \cite[Theorem 4.2]{Khalil2002Nonlinear}, it implies that $r^\beta(t)$ asymptotically converges to $r_d^\beta$, i.e., $\lim_{t\rightarrow \infty} r^\beta(t) = r_d^\beta$.
	\end{proof}

	\begin{figure}[!t]
		\centering{\includegraphics[width=0.7\linewidth]{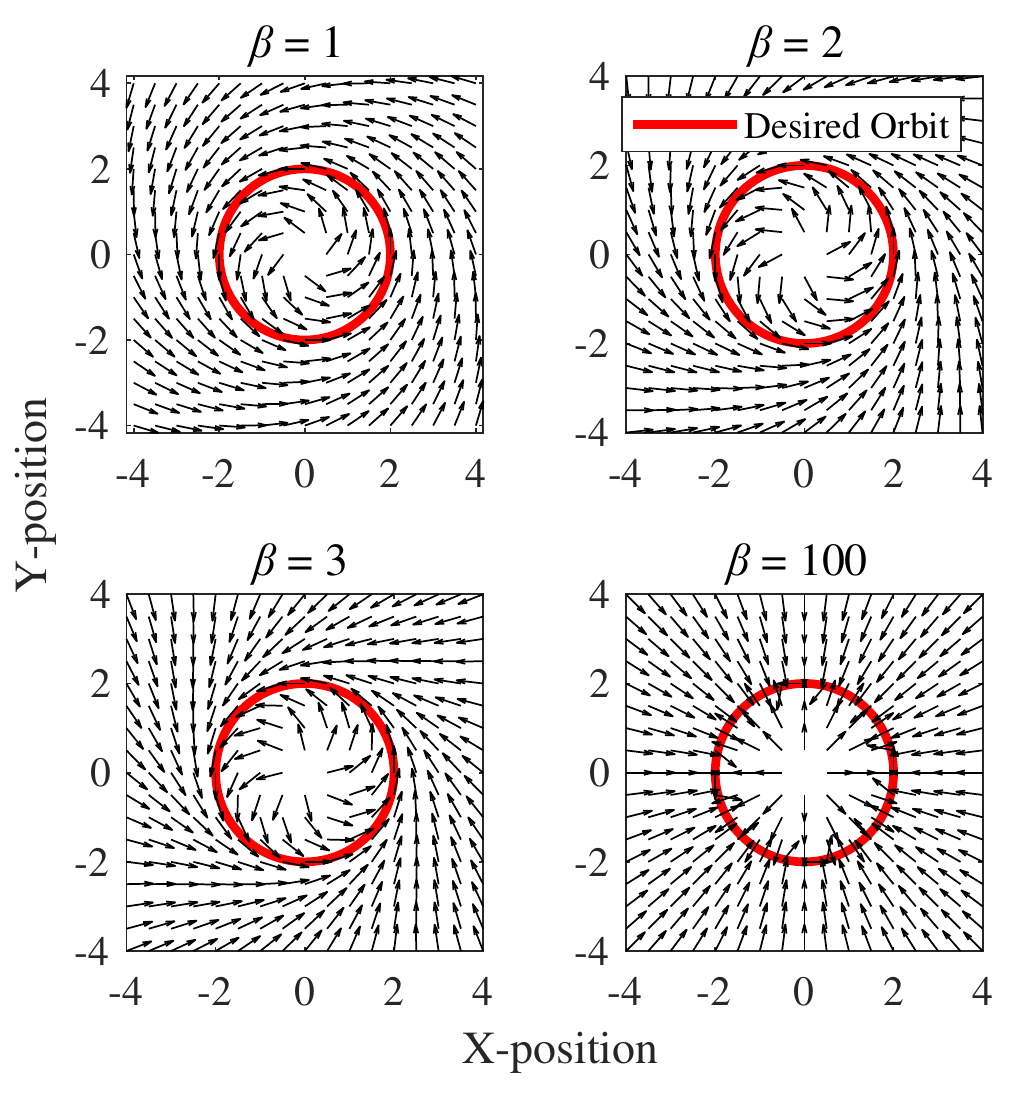}}
		\caption{The LGV fields with different exponents $\beta\in\{1,2,3,100\}$.}
		\label{figmind}
	\end{figure}
	\begin{figure}[!t]
		\centering{\includegraphics[width=0.65\linewidth]{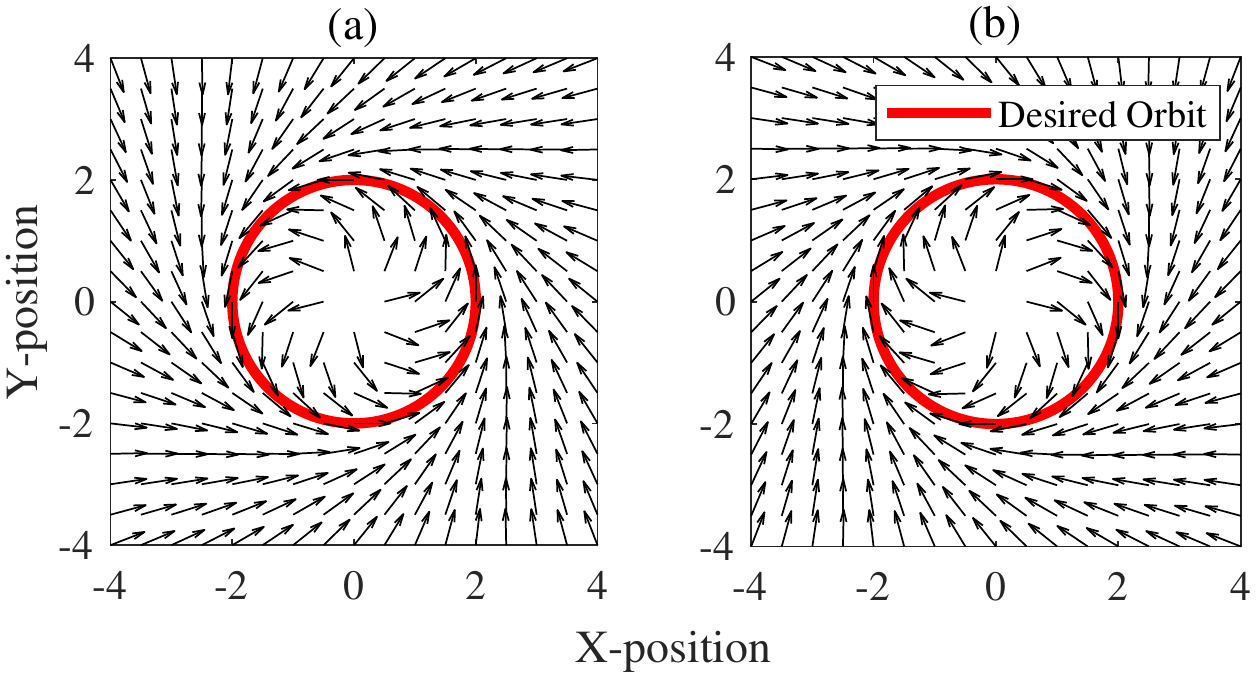}}
		\caption{The LGV fields: (a) the counterclockwise direction;  (b) the clockwise direction.}
		\label{figdirection}
	\end{figure}
	
	If  $r(t)\ge r_d$, which  usually holds for $t=t_0$ in the standoff tracking problem \cite{Matveev2011Range}, it follows from \eqref{eql1} that
	\begin{align}\label{eql11} 
	\dot L_1(r(t)) \le -\left({\beta v_d }/{r_d} \right)L_1 (r(t)).
	\end{align}
	Jointly with \eqref{lyfun}, it implies that the larger $\beta$, the faster the decreasing rate of $r(t)$ to $r_d$. On the other hand, if $r(t)< r_d$, one can easily obtain that $\dot L_1(r(t)) >-\left({\beta v_d }/{r_d} \right)L_1 (r(t)).$ Jointly with Lemma \ref{lemma1} and \eqref{lyfun}, it implies that $r(t)$ will increase to $r_d$ with an exponential rate faster than ${\beta v_d }/{r_d}$.  From this perspective, increasing $\beta$ can speed up the convergence of the LGV, which is also verified by Fig.~\ref{figmind}. Since there are physical input limits for the UAV in \eqref{model}, e.g.,  ${\bm u}(t)\in[0,\bar u]$, the UAV cannot follow an LGV with an arbitrarily large $\beta$.  In fact, $|\dot r(t)|\le v_d$ implies that the convergence rate of the LGV has an upper bound. The fastest one needs manual tuning to match this limit as there is no theoretical result to justify how to achieve it.  Note that there is a trade-off between the convergence rate and the computational burden in \eqref{lya_cont}. In our simulations, we fix $\beta=3$ for an illustration. Interestingly, the case of $\beta=2$ leads to the GV method in \cite{Frew2008Coordinated,pothen2017curvature,Dong2019Flight}. Thus, we not only further justify their design but also generalize it to obtain faster convergence rates.  A discrete-time version of \eqref{lya_cont} with $\beta=2$ is also devised in our previous work \cite{Dong2019Flight}, where the sampling frequency should be faster than an explicit lower bound to ensure its effectiveness.

	The LGV  field in \eqref{lya_cont} will guide the UAV to encircle the target in a counterclockwise direction. If the desired direction is clockwise,  we simply let $\gamma = - \pi/2$ in \eqref{eqnab}  and the LGV is  modified as
	\begin{align} \label{lya_cont1}
	\begin{split}
	\bm v_L(t) =&~ \frac{-v_d}{r(t)\left(r^\beta(t)+r_d^\beta\right) }\\
	& \times \bmatri x(t)\left(r^\beta(t) -r_d^\beta\right) - 2y(t)\sqrt{r^\beta(t) r_d^\beta} \\ y(t)\left(r^\beta(t)-r_d^\beta\right) +2x(t)\sqrt{r^\beta(t) r_d^\beta} \ematri.
	\end{split}
	\end{align}
See Fig.~\ref{figdirection} for an illustration. There is no loss of generality to select any direction, as one can achieve the same results for the other direction.  Note that if the controller continuously  switches its direction, it will potentially destabilize the closed-loop system. 
	
	\subsection{The Justification of our LGV from the geometry perspective}
	\begin{figure}[t!]
		\centering{\includegraphics[width=0.6\linewidth]{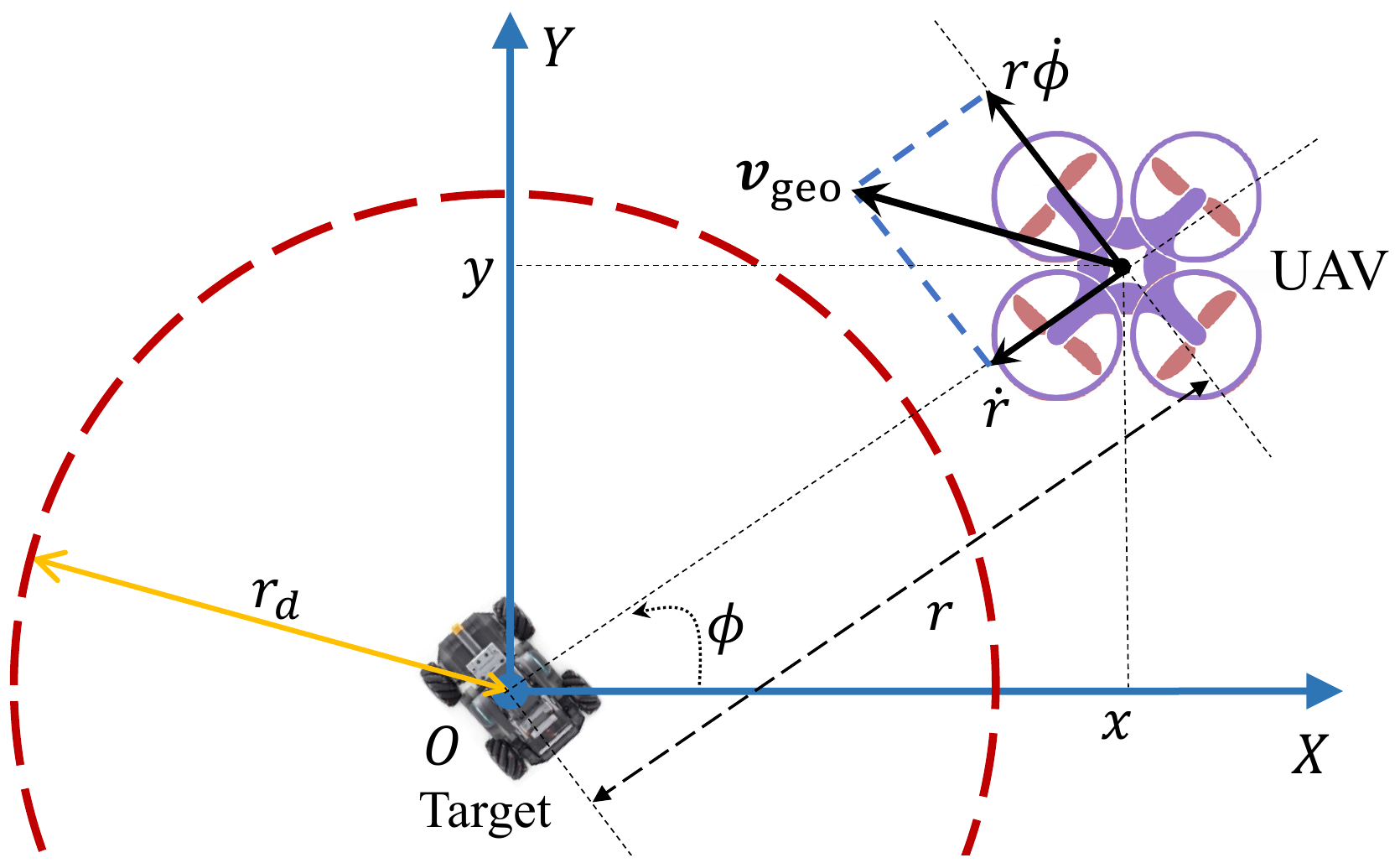}}
		\caption{Coordinate frame centered at the target.}
		\label{figpolar}
	\end{figure}
	We  further justify our LGV in \eqref{lya_cont} from the geometry perspective.  As illustrated in Fig.~\ref{figpolar}, we establish a polar frame centered at the target and denote the coordinates of the UAV by the radius $r(t)$ and angle $\phi(t)$, respectively. Then, their relative position $x(t) = x_q(t) -x_o(t)$ and $y(t)=y_q(t) -y_o(t)$ on the horizon plane are expressed as  
	\begin{align*}
	x(t) = r(t) \cos \phi(t)~ \text{and}~ y(t) = r(t)\sin \phi(t).
	\end{align*}
	Taking the time derivative of $x(t)$ and $y(t)$ leads to that
	\begin{align} \label{eq_convert}
	\bmatri \dot x(t) \\ \dot y(t) \ematri = \bmatri \cos \phi(t) & -\sin \phi(t)  \\ \sin \phi(t) & \cos \phi(t) \ematri  \bmatri \dot r(t) \\ r(t)\dot \phi(t) \ematri
	\triangleq \bm v_{\geo}(t).
	\end{align}
	Thus, we can first design the radial velocity $\dot r(t)$ and tangential velocity $r(t)\dot \phi(t)$ in the polar frame and then convert them into the Cartesian frame by the rotation matrix in \eqref{eq_convert}.
	
	Consider the following Lyapunov function candidate
	\begin{align}
	L_2(r(t)) = \frac{1}{2} \left(r^\beta(t) -r_d^\beta  \right)^2.
	\end{align}
	Taking the derivative of $L_2(r(t))$ yields that
	\begin{align}
	\dot L_2(r(t)) = \beta\left(r^\beta(t) -r_d^\beta\right)r^{\beta-1}(t) \dot r(t).
	\end{align}
	To ensure the negativeness of $\dot L_2(r(t))$, we design the radial velocity $\dot r(t)$ as 
	\begin{align} \label{dot_r}
	\dot r(t) = -\alpha v_d \left(r^\beta(t) -r_d^\beta\right),~ \alpha >0.
	\end{align} 
	It immediately follows from \eqref{dot_r} that \[\dot L_2(r(t))=-\alpha \beta v_d r^{\beta-1}(t) \left(r^\beta(t)-r_d^\beta\right)^2,~ \lim_{t\rightarrow \infty} r^\beta(t) = r_d^\beta.\]
	
	Moreover, let the tangential velocity $r(t)\dot \phi(t)$ satisfy that  
	\begin{align} \label{eqrphi}
	\dot r^2(t) + \left(r(t) \dot \phi(t)\right)^2 = v_d^2.
	\end{align}
	By \eqref{dot_r}, it requires that $\alpha \left(r^\beta(t) -r_d^\beta\right)\le 1$ for any $r(t)> 0$. To this end, we select $\alpha$ as 
	\begin{align} \label{eqalpha2}
	\alpha = \left(r^\beta(t) +r_d^\beta\right)^{-1}.
	\end{align}
	Inserting \eqref{eqalpha2} into \eqref{dot_r} and \eqref{eqrphi} yields the following radial and tangential velocities
	\begin{align} \label{eqdotr}
	\begin{split}
	\dot r(t) = -v_d \frac{ r^\beta(t) -r_d^\beta}{r^\beta(t) +r_d^\beta},~
	r(t) \dot \phi(t)= \pm v_d \frac{ 2\sqrt{r^\beta(t) r_d^\beta}}{r^\beta(t) +r_d^\beta}
	\end{split}.
	\end{align}
	Note that the direction of $\dot r(t)$ always points to the desired circle, i.e., the red one in Fig.~\ref{figgradient}, and that of $r(t) \dot \phi(t)$ needs to be manually specified as either the counterclockwise ``$+$" or clockwise ``$-$" direction. This is consistent with the value selection of $\gamma$ in \eqref{direction}. Combining \eqref{eq_convert} with \eqref{eqdotr}, one can easily obtain the LGV in \eqref{lya_cont} and \eqref{lya_cont1}. 
	
	\subsection{Trajectory planning by the LGV} \label{sub_11}
	
	\begin{figure}[!t]
	\centering{\includegraphics[width=0.7\linewidth]{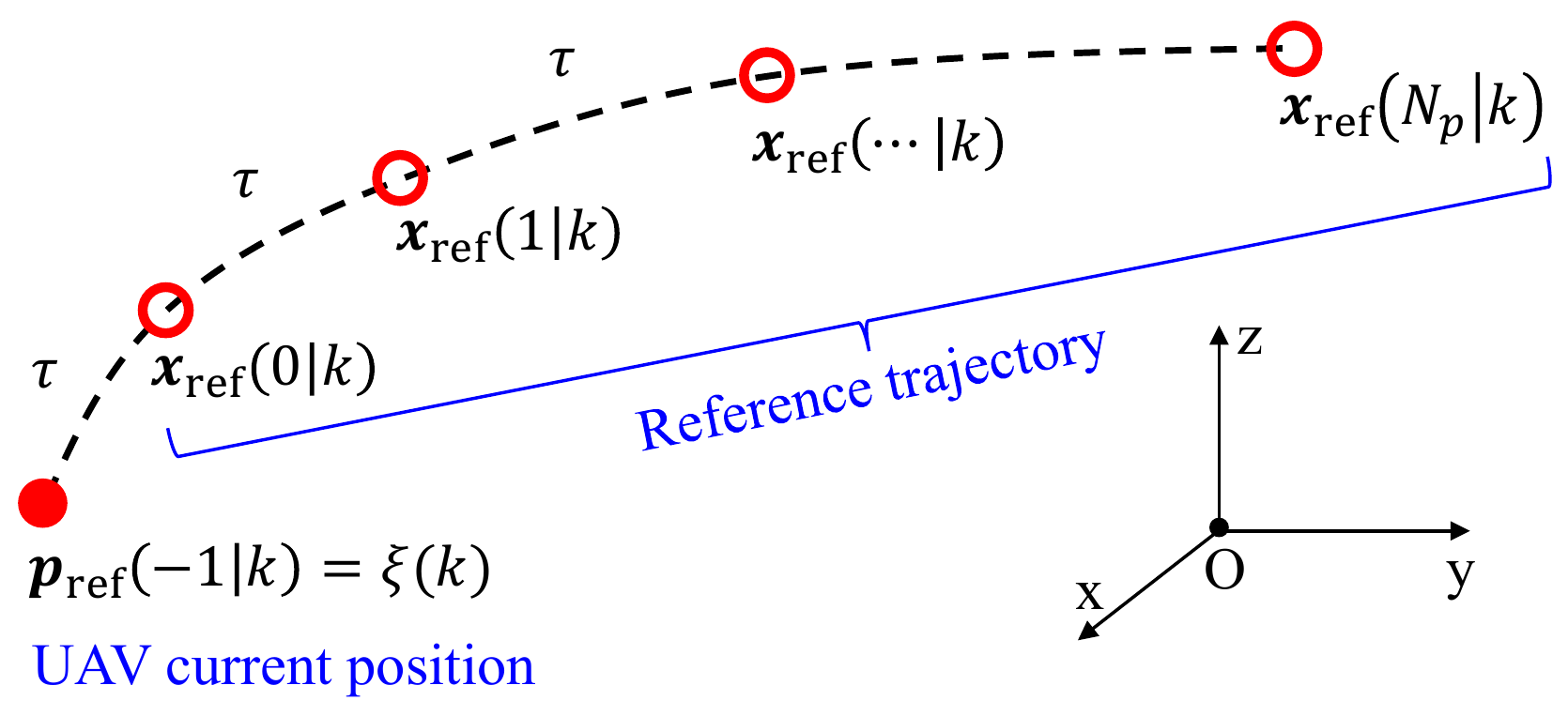}}
	\caption{Illustration of trajectory planning by the LGV.}
	\label{figsplanning}
\end{figure}
	
	 At the UAV position $\bm \xi(k)$, a reference trajectory is planned via the discretized LGV in \eqref{lya_cont} with a sampling period $\tau$, i.e.,
	\begin{align*}
	\bm x_{\r}  (i|k) = \bmatri\bm p_{\r}^{\T}(i|k),\bm 0_{3}^{\T}, \bm v_{\r}^{\T}(i|k), \bm 0_{3}^{\T} \ematri^{\T},~i=0,1,\ldots,N_p
	\end{align*}
	where $\bm p_{\r}(i|k)=[x_{\r}(i|k), y_{\r}(i|k), z_{\r}(i|k)]^{\T}$ and $\bm v_{\r}(i|k)\in \mathbb{R}^{3}$ denote the predicted reference position and velocity at the $i$-th  time-step ahead, respectively. Moreover, the reference Euler angles and rates are always set to zero for the stability consideration of the UAV. 
	
		Accordingly, define the predicted relative position and distance as 
		\begin{align*}
		[x(i|k),y(i|k),z(i|k)]^{\T}&=\bm p_\r(i|k) -\bm p_o(i|k)\\
		r(i|k)&=\left(x^2(i|k)+y^2(i|k)\right)^{1/2}.
		\end{align*} 
		By replacing $x(t)$ and $y(t)$ in \eqref{lya_cont} with $x(i|k)$ and $y(i|k)$, the predicted LGV is naturally given below 
		\begin{align}\label{eq11}
		\begin{split}
		&\bm v_L(i+1|k) = \frac{-v_d}{r(i|k)\left(r^\beta(i|k)+r_d^\beta\right) } \\
		&~~~\times \bmatri x(i|k)\left(r^\beta(i|k) -r_d^\beta\right) + 2y(i|k)\sqrt{r^\beta(i|k) r_d^\beta} \\ y(i|k)\left(r^\beta(i|k)-r_d^\beta\right) -2x(i|k)\sqrt{r^\beta(i|k) r_d^\beta} \ematri.
		\end{split}
		\end{align}
	
	Next, we show how to update the required quantities in \eqref{eq11}.  Firstly, it follows from  \eqref{lya_cont} that the reference velocity $\bm v_{\r}(i+1|k)$   is  updated  as  
	\begin{align} \label{eq_gui}
	\bm v_{\r}(i+1|k)= \bmatri \bm v_L(i+1|k)\\ v_z \cdot \tanh\left(z_d - z(i|k) \right)\ematri + \bm v_o(k),
	\end{align}
where $v_z$ is a predefined vertical speed. Note that the standard hyperbolic tangent function $\tanh(\cdot)$ is adopted to ensure the boundedness of $\Vert \bm v_{\r}(i|k) \Vert$, i.e., $$\Vert \bm v_{\r}(i|k) \Vert\le (v_d^2+v_z^2)^{1/2} + \Vert \bm v_o(k) \Vert.$$ If the UAV is faster than the above upper bound, then   $\bm v_{\r}(i|k)$ is feasible for the UAV.  
	
	Secondly, the reference position is initialized as $\bm p_{\r}(-1|k)=\bm \xi(k)$ (cf. Fig.~\ref{figsplanning}) and updated via  a single-integrator. That is 
	\begin{align} \label{eqpath}
	\bm p_{\r}(i+1|k) = \bm p_{\r}(i|k) + \tau \cdot \bm v_{\r}(i+1|k).
	\end{align} 
	Note that we have shifted $\bm v_\r(\cdot|k)$ one step ahead for notational convenience.
	If $\tau$ is sufficiently small,  it follows from \cite[Theorem 1]{Dong2019Flight} that the reference trajectory $\{\bm x_{\r}(i|k)\}_{i=0}^{N_p}$ will guide the UAV to achieve \eqref{obj}. 
 	
	Thirdly, we simply assume that the target velocity is invariant in the trajectory planning horizon and predict the target position as		
	$${\bm p}_o(i+1|k)= {\bm p}_o(i|k) + \tau {\bm v_o(k)}~\text{with}~ {\bm p}_o(-1|k)=\bm p_o(k).$$
	
	Overall, our trajectory planning is fully described in Algorithm \ref{pathplanning}, where  time indices of  $x, y, z, r$ are omitted for brevity.
	
	\begin{algorithm} [!t]
		\caption{Trajectory planning by the LGV.}
		\label{pathplanning}
		\begin{enumerate}\renewcommand{\labelenumi}{\rm(\alph{enumi})}	
			\item \textbf{Input:} the UAV position $\bm \xi(k)$, the target state $[\bm p_o^\T(k),\bm v_o^\T(k)]^\T$, the desired tracking formation vector $[r_d, z_d, v_d]^\T$. 
			\item \textbf{Output:} the reference trajectory $\{ \bm x_{\r}(i|k)\}_{i=0}^{N_p}$.
			\begin{itemize}
				\item \textbf{initialization:} Let $\bm p_{\r}(-1|k)=\bm \xi(k)$, and ${\bm p}_o(-1|k)=\bm p_o(k)$.
				\item \textbf{Update:} For $i=-1,0,1,\ldots,N_p-1$ 
				\begin{itemize}
					\item Compute the relative position and projected distance:
					\begin{align*}
					[x, y, z]^\T &= \bm p_{\r}(i|k) - {\bm p}_o(i|k), \\
					r&=\sqrt{x^2+y^2}.
					\end{align*}
					\item Compute the predicted LGV   by \eqref{lya_cont}, i.e., 
					\begin{align*}
					\hspace{-1.2cm}\bm v_L(i+1|k) = \frac{-v_d}{r(r^\beta+r_d^\beta) } \bmatri  x(r^\beta -r_d^n) + y(2\sqrt{r^\beta r_d^\beta}) \\ y(r^\beta-r_d^\beta) -x(2\sqrt{r^\beta r_d^\beta}) \ematri.
					\end{align*}
					\item Compute the next desired states by \eqref{eq_gui}:
					\begin{align*}
					\bm v_{\r} (i+1|k) &= \bmatri \bm v_L(i+1|k)\\v_z\cdot \tanh\left( z_d-z\right) \ematri + {\bm v_o(k)},\\
					\bm p_{\r}(i+1|k) &= \bm p_{\r}(i|k) + \tau \bm v_{\r} (i+1|k),\\
					{\bm p}_o(i+1|k) &= {\bm p}_o(i|k) + \tau {\bm v_o(k)}.
					\end{align*}
				\end{itemize}
			\end{itemize}	
		\end{enumerate}			
	\end{algorithm}

	\section{The Design of MPCs} \label{sec_controller}
	In this section, we design both nonlinear and linearized MPCs to  track the reference trajectory $\{ \bm x_{\r}(i|k)\}_{i=0}^{N_p}$ in Algorithm \ref{pathplanning}.      
		
	\subsection{The nonlinear MPC}
	Let $\bm x(t) =\text{col}\left(\bm \xi(t), \bm \eta(t), \dot {\bm \xi}(t), \dot {\bm \eta}(t)\right) $. We rewrite the model in \eqref{model} into the following abstract form
	\begin{align} \label{eqcontinous}
	\dot {\bm x}(t) = f(\bm x(t), \bm u(t)),
	\end{align}
	 and use the explicit Euler method to discretize it
	\begin{align} \label{eqdis}
	{\bm x}(k+1) = {\bm x}(k) + \tau\cdot f(\bm x(k), \bm u(k)).
	\end{align}
	Then, our objective is on the control design to track the reference states $\{ \bm x_{\r}(i|k)\}_{i=0}^{N_p}$.  To this end, define the control input sequence as
	$
	\bm u(\cdot|k)=\left\{\bm u(0|k), \bm u(1|k),\ldots,\bm u(N_p-1|k) \right\} 
	$
	and the objective function as follows
	\begin{align*}
	J(\bm u(\cdot|k)) = &\frac{1}{2} \sum_{i=0}^{N_p} \left(\bm x(i|k) -\bm x_{\r}(i|k)\right)^{\T} Q \left(\bm x(i|k) -\bm x_{\r}(i|k) \right) \\
	&+ \frac{1}{2}\sum_{i=0}^{N_p-1}  \left(\bm u(i|k) - \bm u_{\r} \right)^{\T} R \left( \bm u(i|k) - \bm u_{\r}\right) 
	\end{align*}
	where $\bm u_{\r}=u_\r \cdot {\bm 1}_4^\T$ is a constant input vector for the UAV to ensure that $\bm \eta(t)=\dot{\bm \xi}(t)=\ddot{\bm \xi}(t)=0$ in \eqref{model}, leading to $u_\r= m g/(4l_c)$. Moreover, $Q$ and $R$ are positive semi-definite and positive definite weighting matrices, respectively.

	The nonlinear MPC law is obtained by solving the constrained optimization problem per time-step $k$ viz
	\begin{subequations} \label{nmpc}
		\begin{align}
		&\bm u_*(\cdot|k)= \argmin_{ \bm u(\cdot|k)} J(\bm u(\cdot|k))  \\
		&~~~\mathrm{s.t.}~  \bm x(0|k) =  \bm x(k)   \label{initial},  	      \\
		&~~~~~~~~ \bm x(i+1|k) = \bm x(i|k) + \tau f(\bm x(i|k), \bm u(i|k)  )  \label{dyna},  \\
		&~~~~~~~~ \phi(i|k), \theta(i|k) \in [-c, c]  \label{xcons}, \\
		&~~~~~~~~\bm u(i|k)  \in [0, \bar u] \label{cons}, \\
		&~~~~~~~~\bm u(j+1|k) = \bm u(j|k), ~\forall j \ge N_c-1,
		\end{align}
	\end{subequations}
	where \eqref{xcons} represents the  constraint on pitch and roll angles of the UAV, and \eqref{cons} denotes the input constraint. In addition, $N_p$ and $N_c$ denote the tunable prediction and control horizons. To reduce the numerical complexity of \eqref{nmpc}, we omit the terminal  cost and constraint set  \cite{gros2020linear,cimini2020embedded} whose effect diminishes as the horizon length grows.

	\subsection{The linearized MPC}
	Since it is hard to directly solve the NLP in \eqref{nmpc}, one may consider to converting into a quadratic program (QP) by linearizing \eqref{eqdis} at the current state $\bm x(k)$ as 
	\begin{align*} 
	\bm x(k+1) = A_k \bm x(k) + B_k \bm u(k),
	\end{align*}
	where
	\begin{align*}
	A_k = I_{12} +\tau \frac{\partial f( \bm x(k),\bm u(k))}{\partial \bm x}~\text{and}~ B_k =\tau \frac{\partial f( \bm x(k),\bm u(k))}{\partial \bm u}.
	\end{align*}
Then, the nonlinear model in \eqref{dyna} is approximated by the following linearized prediction model
	\bee\label{linear}
	\bm x(i+1|k) = A_k \bm x(i|k) + B_k \bm u(i|k).  
	\ene
Note that other solutions can also be applied to limit the effect of prediction errors,  e.g., linearization around a trajectory or robust approaches \cite{rawlings2017model}.

	By letting $\Delta \bm x(i|k) =  \bm x(i|k) - \bm x_{\r}(i|k)$ and $
	\Delta \bm u(i|k) =  \bm u(i|k) - \bm u_{\r}$, 
	we transform \eqref{linear} into a LIT system, i.e., 
	\begin{equation} \label{error}
	\Delta \bm x(i+1|k) = A_k \Delta \bm x(i|k) + B_k \Delta \bm u(i|k) + \bm r(i+1|k)   
	\end{equation} 
	where $\bm r(i+1|k)$ is a nonzero offset and is given by 
	\begin{align*}
	\bm r(i+1|k) = A_k \bm x_{\r}(i|k) + B_k \bm u_{\r} - \bm x_{\r}(i+1|k).
	\end{align*}
	We define
	\begin{align*}
	\Delta X &= \text{col}\left(\Delta \bm x(1|k),\ldots,\Delta \bm x(N_p|k)\right), \\
	\Delta U &= \text{col}\left(\Delta \bm u(0|k), \ldots, \Delta \bm u(N_c-1|k)\right),
	\end{align*}
	and write \eqref{error} into a collected vector form
	\begin{align}\label{compact}
	\Delta X =  \wt{A}_k \Delta\bm x(0|k) + \wt{B}_k \Delta U + \wt{D}_k,
	\end{align}
	where $\wt{A}_k$ and $\wt{B}_k$ are given in \eqref{wAB}, and $\wt{D}_k$ is obtained by
	\begin{align*}
	\wt{D}_k= \bmatri \bm r(1|k) \\ A_k\bm r(1|k) + \bm r(2|k)\\ \vdots \\A_k^{N_p-1}\bm r(1|k)+A_k^{N_p-2}\bm r(2|k) +\ldots+ \bm r(N_p|k)   \ematri.
	\end{align*}
	\newcounter{mytempeqncnt}
	\begin{figure*}[!t]	
		\normalsize
		\setcounter{mytempeqncnt}{\value{equation}}		
		\begin{align} \label{wAB}
		\wt{A}_k = \bmatri  A_k \\ A_k^2\\ \vdots\\ A_k^{N_p} \ematri, ~~
		\wt{B}_k = \bmatri B_k & 0 & \cdots & 0 \\ 
		A_k B_k & B_k & \ldots & 0 \\
		\vdots&\vdots&\ddots&\vdots \\ 
		A_k^{N_c-1} B_k & A_k^{N_c-2} B_k &\cdots & 	 B_k\\
		A_k^{N_c} B_k & A_k^{N_c-1} B_k &\cdots & 	A_k B_k +B_k\\ 
		\vdots&\vdots&\ddots&\vdots \\
		A_k^{N_p-1} B_k & A_k^{N_p-2} B_k &\cdots &A_k^{N_p-N_c} B_k+\cdots+ A_k B_k +B_k\ematri.
		\end{align}
		\hrulefill	
	\end{figure*}
	Letting $C \bm x(i|k) \le \bm c$ represent the state constraint in \eqref{xcons} and $\text{blkdiag}(\cdot)$ be a block diagonal matrix, the state constraint is rewritten as 
	\begin{align} \label{eq_cc}
	\wt{C}  \wt{B}_k \Delta U \le 
	\wt{\bm c} -\wt{C} \wt{A}_k \Delta\bm x(0|k)-\wt{C} \wt{D}_k,
	\end{align}
	where $\wt{C} = \text{blkdiag}\left(C, \ldots, C\right)$ and   
	\begin{align*}
	\wt{\bm c} = \text{col}\left( \bm c, \ldots, \bm c \right) - \wt{C}  
	\text{col}\left( \bm x_{\r}(1|k), \ldots, \bm x_{\r}(N_p|k) \right).
	\end{align*}
	
	By \eqref{compact} and \eqref{eq_cc}, the NLP in \eqref{nmpc} is converted into the following QP problem
	\begin{subequations} \label{qp}
		\begin{align}
		&\Delta U_*  =  \argmin_{ \Delta U} \frac{1}{2} \Delta U^{\T} H \Delta U + \bm h^{\T} \Delta U \label{qpobj}\\
		&~~~\mathrm{s.t.}~~  W \Delta U \le \bm w. \label{eq_cons} 
		\end{align}
	\end{subequations}
	Here the new weighting matrix $H$ and vector $\bm h$ are given by 
	\begin{align*} 
	H = \wt{B}_k ^{\T} \wt{Q}  \wt{B}_k +\wt{R}, ~\bm h = \wt{B}_k^\T \wt{Q} ^\T \left(\wt{A}_k \Delta\bm x(0|k) + \wt{D}_k\right),
	\end{align*}
	where  
	$
	\wt{Q} = \text{blkdiag}(Q,\ldots,Q),~\wt{R} = \text{blkdiag}(R,\ldots,R)$  
	and the matrix $W$ and vector $\bm w$ in constraint \eqref{eq_cons}  are given by
	\begin{align*}
	W = \bmatri \wt{C} \wt{B}_k\\ I_4\\ -I_4 \ematri, ~~
	\bm w = \bmatri \wt{\bm c} -\wt{C} \wt{A}_k \Delta \bm x(0|k)-\wt{C} \wt{D}_k\\
	(\bar{u}-u_{\r})\bm{1}_{4}\\
	u_{\r}\bm{1}_{4} \ematri.
	\end{align*}
		\subsection{Numerical comparison between the nonlinear and linearized MPCs in MATLAB} \label{numberical}
		\begin{table*}[t!]
		\centering
		\caption{Comparison between the nonlinear and linearized MPCs on Intel Core i5-6500 CPU@$3.2$\si{GHz}.}
		\begin{tabular}{cc:cc:cc:cc:cc}
			\toprule[1pt]
			%		\hline
			\multicolumn{2}{c:}{Horizon length} & \multicolumn{2}{c:}{Solution time (s)} & \multicolumn{2}{c:}{Range error (m)} & \multicolumn{2}{c:}{Speed error (m/s)} & \multicolumn{2}{c}{Height error (m)} \\
		 \multicolumn{1}{c}{$N_p$}   & \multicolumn{1}{c:}{$N_c$}     & \multicolumn{1}{c}{NMPC} & \multicolumn{1}{c:}{LMPC} & \multicolumn{1}{c}{NMPC} & \multicolumn{1}{c:}{LMPC} & \multicolumn{1}{c}{NMPC} & \multicolumn{1}{c:}{LMPC} & \multicolumn{1}{c}{NMPC} & \multicolumn{1}{c}{LMPC} \\
			\hline
			$10$    & $5$     & $0.0903$ & $0.0067$   & $0.1144$ & $0.0420$ & $0.0076$ & $0.0101$ & $0.0013$ & $0.0091$ \\
			\hline
			$10$    & $10$    & $0.1712$ & ${0.0070}$   & $0.1115$ & $0.0449$ & ${\bf 0.0075}$ & $0.0099$ & $0.0012$ & $0.0093$ \\
			\hline
			%		\specialrule{0em}{1pt}{1pt}
			${\bf 20}$    & ${\bf 10}$    & $0.3778$ & $ 0.0089$   & $0.0165$ & $0.0209$ & $0.0114$ & $0.0250$ & $1.71\times10^{-4}$ & $0.0090$ \\
			\hline
		 ${20}$    & ${20}$    & ${0.6481}$ & ${0.0100}$  & $0.0180$ & $0.0264$ & $0.0111$ & $0.0242$ & $\bm {1.28\times10^{-4}}$ & $0.0090$ \\
			\hline
			$40$    & $20$    & $2.1852$ & $0.0116$  & ${\bf 0.0124}$ & $0.0534$ & $0.0177$ & $0.0258$ & $5.21\times10^{-4}$ & $0.0090$ \\
			\hline
			$40$    & $40$    & $2.6402$ & $0.0145$ & $0.0127$ & $0.0573$ & $0.0174$ & $0.0235$ & $5.27\times10^{-4}$ & $0.0090$ \\	
			%		\hline
			%		\bottomrule
			\bottomrule[1pt]
		\end{tabular}%
		\label{tabl}%
	\end{table*}%
	
	For comparing the nonlinear MPC in \eqref{nmpc} with its linearized counterpart in \eqref{qp}, let $Q=\diag(\bm 1_9,\bm 0_3)$, $R=0.1\cdot I_4$, $c=\pi/4$, $\bar u = 12$ and the standoff tracking pattern is specified by  $r_d = 2.0$\si{m}, $z_d=5.0$\si{m}, $v_d = 1.0$\si{m/s}.  In this simulation, they are solved by directly invoking  \texttt{nlmpcmove}   and  \texttt{quadprog}   in MATLAB 2020a on the personal computer with an Intel Core i5-6500 CPU@$3.2$\si{GHz}.
	The steady-state tracking results are reported in Table \ref{tabl}. Though the nonlinear MPC indeed outperforms the linearized one,  their steady-state tracking performances are not significantly different for the horizon length up to $40$ ($=4$ seconds).  We only report the transition results in Figs. 12-16 for their DNN-based versions for the purpose of saving space. 
 	
	\begin{table*}[t!]
		\centering
		\caption{Solution time comparison between the nonlinear, linearized and DNN-based MPCs with $N_p=20$, $N_c=10$, $\tau=0.1$\si{s}.}
		\begin{tabular}{c:ccc:ccc:c}
			\toprule[1pt]
			%		\hline
			\multicolumn{1}{c:}{Platform}  & \multicolumn{3}{c:}{Intel Core i5@$3.2$\si{GHz}} & \multicolumn{3}{c:}{ARM@$200$\si{MHz}} & \multicolumn{1}{c}{FPGA@$200$\si{MHz}} \\
			\hline
			{Methods}       
			& {NMPC} & {LMPC} & {DNN-based MPC} & {NMPC} & {LMPC} & {DNN--based MPC} & {DNN--based MPC}  \\
			\hline
			{Solution time (s)} & $0.3778$  & $0.0089$ & ${\bf 2.81\times 10^{-5}}$ & $120.9$  & $2.8$ & ${\bf 1.03\times10^{-4}}$ & ${\bf 1.26\times10^{-4}}$  \\
			%		\hline
			\bottomrule[1pt]	
		\end{tabular}%
		\label{tab2}%
	\end{table*}%
	
	\section{The DNN-based MPC with an IM} \label{sec5}
		
	\begin{figure}[!t]
		\centering{\includegraphics[width=0.8\linewidth]{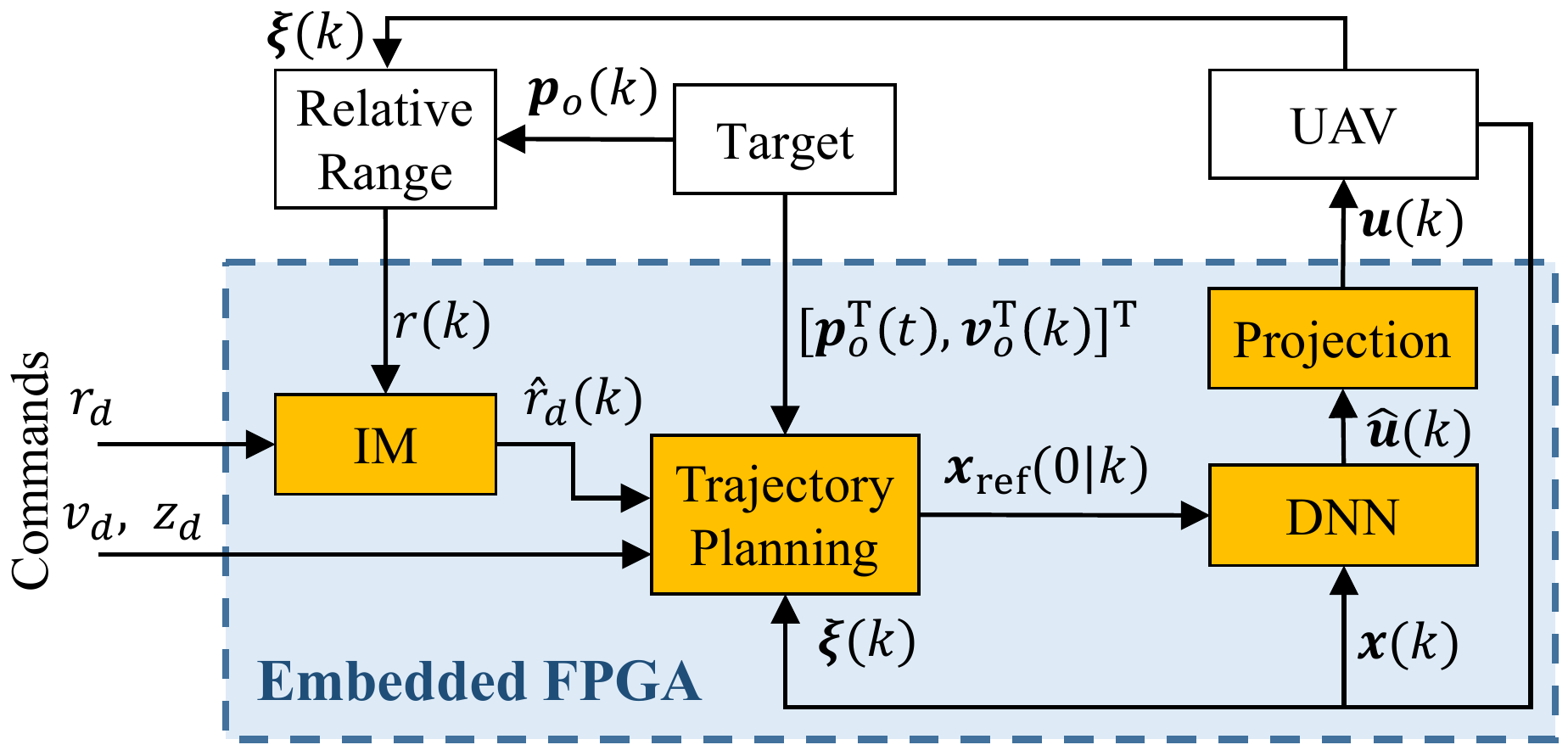}}
		\caption{The standoff tracking system closed by the embedded trajectory planning and DNN-based MPC with an IM on an FPGA.}
		\label{figpinn}
	\end{figure}

Although the solution efficiency of the QP in \eqref{qp} is significantly improved over the NLP in \eqref{nmpc},  it is still challenging for an embedded processor, e.g., ARM or FPGA. For example,  it follows from \cite{domahidi2012efficient} and Eq.~(9) in \cite{hartley2013predictive} that it requires at least $2.5\times 93.58=233.95$\si{ms} for an ARM@$200$\si{MHz} to solve an  QP in \eqref{qp} with $\mathrm{dim}(\bm x)=12$, $\mathrm{dim}(\bm u)=4$, and $N_p=N_c=20$, which is impossible to complete within a sampling period of $0.1$\si{s}. In fact, Table \ref{tab2} also implies that the MPC law is hard to solve on an ARM@$200$\si{MHz} by Eq. (9) in \cite{hartley2013predictive}. Thus, the DNN-based MPC seems to be the only feasible option to solve \eqref{nmpc} or \eqref{qp} on such a resource-limited embedded platform.

In this section, we design a novel DNN-based MPC with an integral module (IM) for the  embedded implementation via an  FPGA and obtain a standoff tracking system shown in Fig.~\ref{figpinn}. In specific, we use the supervised learning to train a DNN-based policy for the implementation of MPC on an embedded FPGA. Then,   an integral module (IM) is designed to refine the tracking performance of the DNN-based MPC.  

	\subsection{The offline sample collection and DNN training}

Clearly, the performance of the DNN-based MPC depends heavily on the quality of training samples. To promote ``richness" of our training samples, we only collect the first limited number of samples per randomized initial state in the closed-loop simulations of Section \ref{numberical}. On one hand, if we only use the open-loop solution, the resulting control input usually takes value in the boundary of its feasible set, e.g. $0$ or $\bar u$ of \eqref{cons}, which intuitively is not informatively rich. On the other hand, the value of the control input in the closed-loop system usually tends to be stationary, and thus the samples in a relatively large time-step are not informatively rich as well. In this work, such a number is determined by the observation on the input fluctuations of the closed-loop system.

In particular, we firstly randomize the state vector $\bm x(k)$ in \eqref{initial} (or \eqref{eq_cons})  and  solve the associated NLP (or QP)  for the nonlinear (or linearized) MPC.  The first element $\bm u_*(k)$ of the optimal solution is added via the pair $\{\bm s(k), \bm u_*(k)\}$ to the training dataset $\mD$, where $\bm s(k) \in \mathbb{R}^{16}$ consists of the UAV state $\bm x(k)$ and the first reference state $\bm x_\r(0|k)$, i.e.,   
\begin{align*}
\bm s(k) = \text{col}\left(\bm x(k), \bm p_{\r}(0|k), \bm v_{\r}(0|k)\right).
\end{align*} 
To reduce the dimension of $\bm s(k)$, we remove zero elements of $\bm x_{\r}(0|k)$.  Then, we keep sampling for a limited number of time steps under this randomized state vector and start over again with a newly randomized  state vector.

	\begin{figure}[!t]
		\centering{\includegraphics[width=0.9\linewidth]{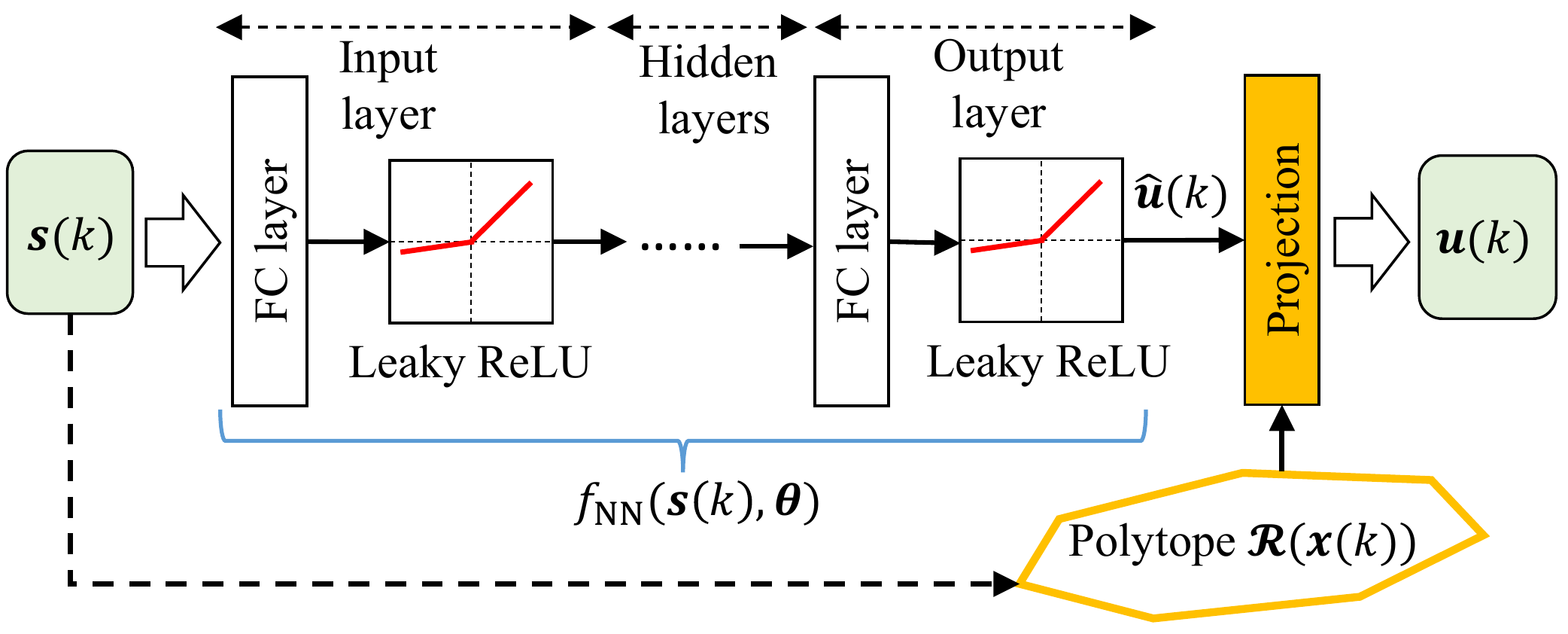}}
		\caption{The DNN architecture and the projection onto $\mR(\bm x(k))$ in \eqref{eqpro}.}
		\label{figNN}
	\end{figure}
	
	We adopt a fully-connected DNN $f_{\NN}(\bm s(k),\bm \theta)$ in Fig.~\ref{figNN}  to fit our samples where $\bm s(k)$ and $\widehat{\bm u}(k)$ are the input and output of the network, and $\bm \theta$ denotes the parameters to be trained, and the leaky rectified linear unit (ReLU) is selected as the nonlinear activation function. That is, the DNN is trained to fit the training data
	\begin{align} \label{train}
	\min\limits_{\bm \theta} \frac{1}{N_\mD} \sum_{i=1}^{N_\mD}\Vert f_\NN(\bm s(i), \bm \theta) - \bm u_*(i) \Vert^2,   
	\end{align}
	where $N_\mD$ is the size of training dataset $\mD$. We adopt the Adam stochastic gradient descent method \cite{kingma2014adam} to solve \eqref{train} with PyTorch, and denote the resulting optimal solution by $\bm \theta_*$.

 Note that both the sample collection and the DNN training are performed {\em offline}. 
 
\subsection{The online projection and integral module to reduce tracking errors} 
To satisfy the state constraint in \eqref{xcons} and the input limit in \eqref{cons}, we further project the DNN output onto the following polytope of the feasible set 
	\begin{align} \label{eqpro}
	\mR (\bm x(k)) = \left\{\bm u \big| C B_k \bm u \le \bm c - C A_k \bm x(k), \bm u\in[0,\bar u]\right\}. 
	\end{align}
In this work, we implicitly assume that $\mR(\bm x(k))$ is non-empty, the recursive feasibility of which has been extensively studied in the MPC theory \cite{rawlings2017model}. 
Since $\mR(\bm x(k))$ in \eqref{eqpro} is the intersection of a finite number of half-spaces,  $\bm u(k)$ can be obtained via a relatively simple QP, i.e.,
\begin{align} \label{qp1}
&\bm u(k) = \text{argmin}_{\bm u} \Vert f_\NN(\bm s(k), \bm \theta_*) - \bm u \Vert^2 \\
&~~~~\mathrm{s.t.}~~  \bm u  \in \mR (\bm x(k)). \notag
\end{align}
By \cite{chen2018approximating}, the projection $\text{PROJ}(\bm u, j)$ of $\bm u$ onto the $j$-th halfspace $C_j B_k \bm u \le c - C_j A_k \bm x(k)$ has a closed form solution, where $C_j$ denotes the $j$-th row of $C$ and $c$ is given in \eqref{xcons}. That is 
\begin{align*}
\text{PROJ}(\bm u, j) = \begin{cases}
\bm u, ~~~~~~~~~~~~\text{if}~C_j B_k \bm u \le  c - C_j A_k \bm x(k),\\
\bm u + \frac{\left(c - C_j (A_k \bm x(k) + B_k \bm u)\right)B_k^\T C_j^\T}{\Vert C_j B_k \Vert^2 }, ~\text{otherwise}.
\end{cases}
\end{align*}
Thus, the QP in \eqref{qp1} can be easily solved by alternating projections. As the  DNN output is ``close" to an optimal solution of the NLP (or QP), any infeasible DNN output is also expected close to the boundary of $\mR (\bm x(k))$, leading to that the computational time of solving \eqref{qp1}   is practically negligible. Instead, we can also simply use a feasible point in $\mathcal{R}(x_k)$ as a backup. 	
	
	Clearly, the approximation error of DNN depends on its size, e.g., width and depth. Since the hardware resources of the embedded chip limit the network scale, the approximation error is inevitable in implementation. In this subsection, we further introduce an IM to refine the tracking performance of the DNN-based MPC. Specifically, the constant command $r_d$ in Algorithm \ref{pathplanning} is replaced by the new time-varying command 
	 \begin{align} \label{eqim}
		\widehat{r}_d(k)= r_d - c_1 \sigma(k),~
		\sigma(k)=\sum_{i=0}^{k} \sat \left( \frac{r(i)-r_d}{c_2}\right),
		\end{align}
		where  $c_i$, $i\in\{1,2\}$ are positive constants and $\sat(\cdot)$ is the standard saturation function \cite{dong2019Target}, i.e.,
		\begin{align*}
		\sat(r)= \begin{cases}
		1, & \text{if}~ r>1,\\
		r, & \text{if}~ \vert r\vert \le 1, \\
		-1, & \text{if}~ r<-1.
		\end{cases}
		\end{align*}
		
		 Unfortunately, it is impossible to provide a rigorous analysis of its impact on the stability due to the high nonlinearity and the inclusion of the DNN in the closed-loop system. We can only informally explain the function of \eqref{eqim}. Suppose that the DNN-based MPC returns a constant steady-state tracking error $b$, i.e., 
	$r(k) \approx r_d + b.$ Then, the localized motion of the closed-loop system in Fig.~\ref{figNN} is approximately given below
		\begin{align} \label{eq_dy}
		r(k+1) = \widehat{r}_d(k) + b.
		\end{align}

		\begin{lemma} \label{lemma3}
			Consider the closed-loop system in \eqref{eq_dy} with a constant bias $b$. If $0<c_2< c_1<2 c_2$, then $r(k)$ asymptotically converges to $r_d$. 
		\end{lemma}
		\begin{proof}
			Substituting \eqref{eqim} into \eqref{eq_dy} leads to that
			\begin{align} \label{eq_r}
			r(k+1) &= r_d - c_1\sigma(k) + b \nonumber\\
			&= \underbrace{r_d - c_1 \sigma(k-1) + b}_{=~r(k)} - c_1 \sat \left( \frac{r(k)-r_d}{c_2}\right) .
			\end{align}
			Let $\delta(k) = r(k)-r_d$. It follows from \eqref{eq_r} that 
			\begin{align}\label{eq_delta}
			\delta(k+1) = \delta(k) - c_1 \sat \left({\delta(k)}/{c_2}\right).
			\end{align} 
			
			If $\delta(k) \ge c_2$, then $\sat(\delta(k)/c_2)=1$ and $\delta(k+1) = \delta(k) - c_1<0$, i.e., $\delta(k)$ decreases to $c_2$. Similarly, if $\delta(k) \le -c_2$, then $\delta(k)$ increases to $-c_2$.
			
			If $\delta(k)\in(-c_2, c_2)$, then $\sat(\delta(k)/ c_2) = \delta(k)/ c_2$ and $\delta(k+1)=(1-c_1/c_2)\delta(k)$. Consider a Lyapunov function candidate $L_3(k) = \delta^2(k)$. Taking its time difference along with \eqref{eq_delta} leads to that 
			\begin{align*}
			\Delta L_3(k) = L_3(k+1)- L_3(k) 
			= -\frac{c_1}{c_2}\left(2-\frac{c_1}{c_2}\right)\delta(k) 
			\le 0.  
			\end{align*}
			By \cite[Theorem 4.2]{Khalil2002Nonlinear}, then $\delta(k)$ asymptotically converges to zero, i.e., $\lim_{k\rightarrow\infty} r(k) = r_d. $
		\end{proof}
		
By Lemma \ref{lemma3}, the IM in \eqref{eqim} aims to eliminate the steady-state range error $b$, which is also verified in Fig.~\ref{figPRS} of Section \ref{sec6e}. However, as the PID, \eqref{eqim} may result in the integral windup when $\delta(k_0)$ is very large.  In practice, one can either adopt \eqref{eqim} only near $r_d$ or select a sufficiently small $c_2$.

	\subsection{An easily accessible framework for the FPGA implementation}
	\begin{figure}[!t]
		\centering{\includegraphics[width=1.0\linewidth]{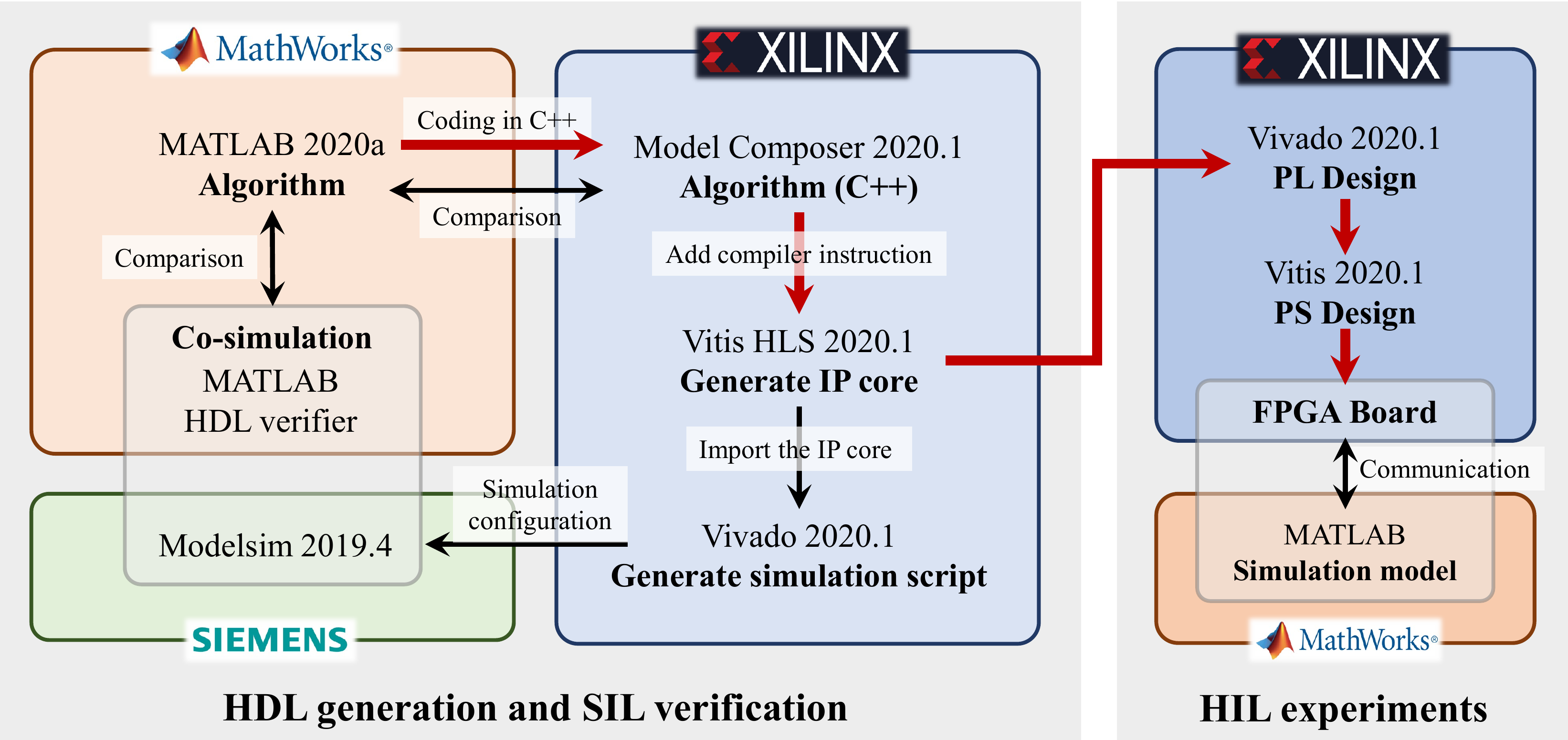}}
		\caption{The HDL generation and FPGA implementation workflow.}
		\label{figWorkflow}
	\end{figure}
	
	The FPGA is good at the parallel computing, and can dramatically accelerate the evaluation of a DNN \cite{nadales2022efficient}. In this work, we introduce an easily accessible framework to show how to  implement our tracking algorithm on an FPGA. The details are given  in Fig. \ref{figWorkflow} with the hardware description language (HDL). The  high-level synthesis (HLS) tool is used to convert the HL programming language for the efficient HDL code directly.  The workflow of the HDL generation is described as follows. 
	\begin{enumerate}
		\item[(a)] Design the tracking algorithms and test them in MATLAB. In this work, they include the trajectory planning in Algorithm \ref{pathplanning}, the DNN in \eqref{train}, the IM in \eqref{eqim}, and the projection in \eqref{qp1}.
		\item[(b)] Code the algorithms of Step (a) in C++ language manually, and then embed into a simulation model through the Xilinx Model Composer to confirm that the C++ code is correct.
		\item[(c)] Use the Vitis HLS tool to generate an IP core described in HDL from the C++ code with the compilable instructions, such as port settings, control protocol, pipeline, loop unwinding, etc. 
		\item[(d)] Generate an IP core in \texttt{verilog} codes, which are then embedded into the HDL simulator (e.g. Modelsim) for correctness checking. 
		\item[(e)] Generate the MATLAB-Modelsim co-simulation module for software-in-the-loop (SIL) simulation.
	\end{enumerate}
	
	The versions of the used softwares are mentioned specifically in Fig.~\ref{figWorkflow}. Following  Steps (a)-(e), we can confirm that the IP core correctly represents our tracking algorithms. Note that the C++ code designed in Step (b) can be directly implemented onto an embedded ARM processor for performance validation.  Once the IP cores are consistent with our Matlab algorithms, we embed them into an FPGA to complete our HIL simulations.

	\section{HIL Simulations} \label{sec6}
	
	In this section, we perform HIL simulations to validate the effectiveness of the proposed DNN-based MPC with an IM on an FPGA@200MHz.

	\subsection{The simulation setting} \label{sec6_3}
		
	\begin{figure}[!t]
		\centering{\includegraphics[width=0.8\linewidth]{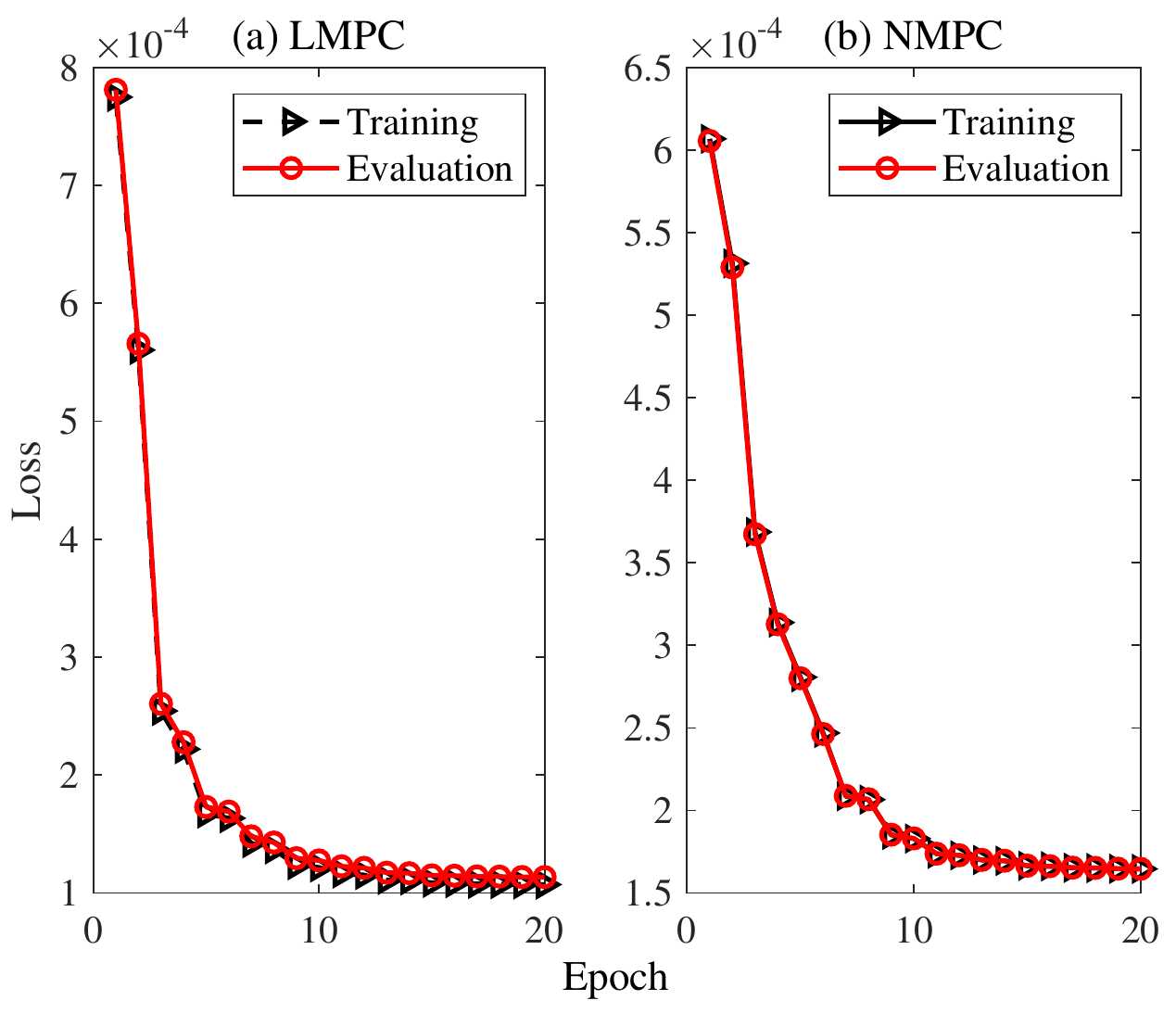}}
		\caption{The mean losses during the training and evaluation process: (a) DNN-based LMPC and (b) DNN-based NMPC.}
		\label{figloss}
	\end{figure}
	We design a fully-connected DNN with $2$ hidden layers and $100$ neurons per layer.  The leaky ReLU is used with a small slope of $0.01$ as activation functions.
	We compare the DNN-based nonlinear and linearized MPCs on the same FPGA with $N_p=20$ and $N_c=10$.  Though the offline training and online implementation of a DNN are the same for both cases, the computational cost of generating samples in the nonlinear one is much more expensive. Specifically, we have generated $10^7$  samples offline for each case, meaning that the number of optimization problems required to solve in each case is up to $10^7$, which takes $27$ hours for the linearized MPC and (estimated) $43.7$ days for the nonlinear MPC  in Section \ref{numberical}. Thus, we  adopt a workstation with an AMD EPYC 7742 CPU having $64$ cores  to generate samples for the nonlinear MPC which takes 205 hours ($\approx 9$ days).  In fact, we have tried the sample size of $5\times 10^6$ and unfortunately, such a sample size is not very satisfactory. Though it is beyond the scope of this work, how to increase the sample efficiency is worthy of further investigation.

	Then, we pick $90\%$ of the samples for training and the rest for evaluation.  Two DNNs are trained for $20$ epochs to generate the corresponding DNN-based NMPC and DNN-based LMPC and Fig.~\ref{figloss} shows the process of training and evaluation by using the Adam \cite{kingma2014adam} with a batch size of $200$. Since the fitting losses reduce marginally after the $15$-th epoch, we regard the parameter after the $20$-th epoch as the optimal solution of \eqref{train}. The training takes around $2.4$ hours on our personal computer.
	
	\subsection{The HIL simulations on an FPGA@$200$\si{MHz} } \label{sec_c}
	
	\begin{figure}[!t]
		\centering{\includegraphics[width=0.6\linewidth]{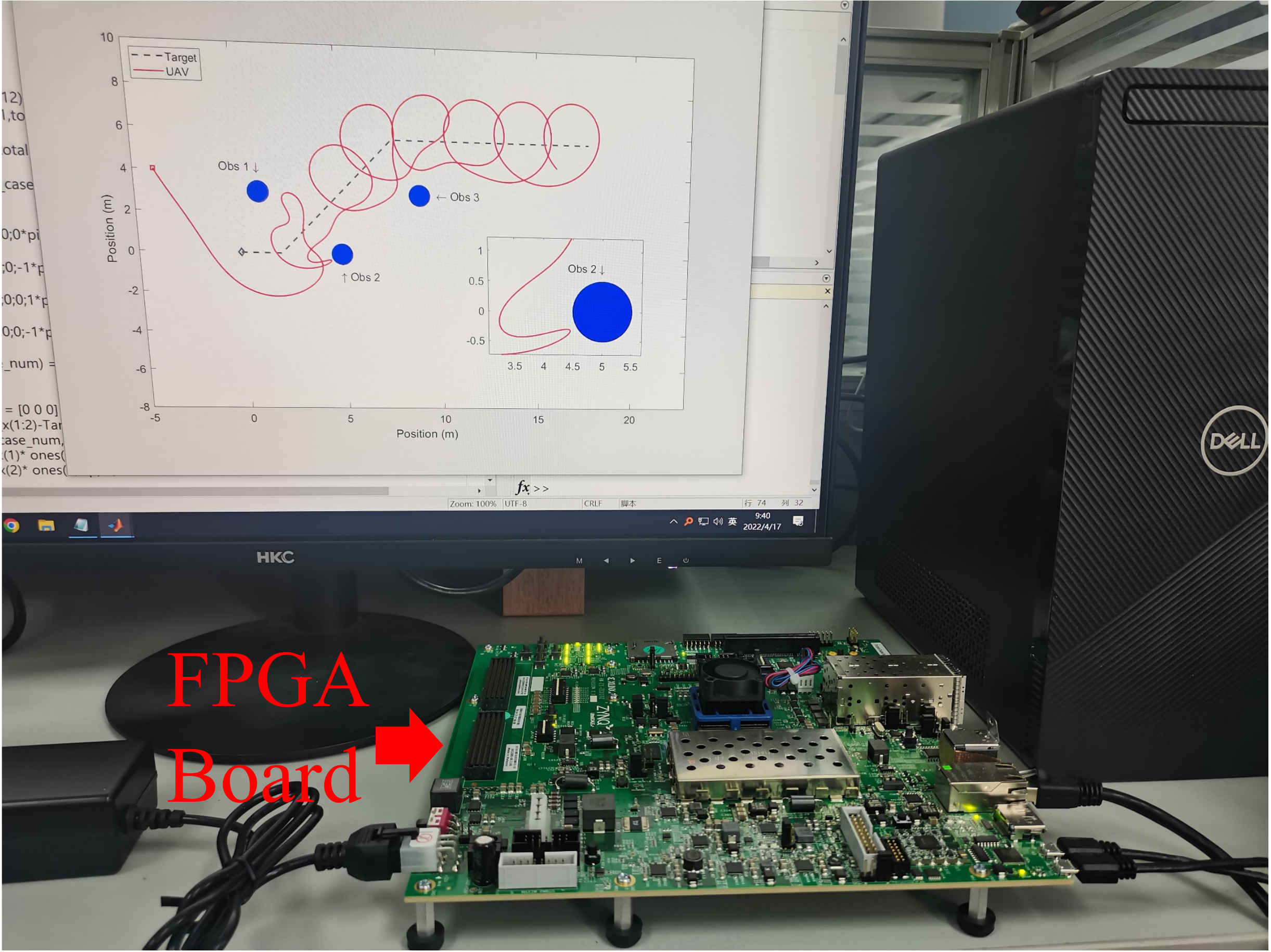}}
		\caption{The Xilinx ZCU102 evaluation board.}
		\label{figfpga}
	\end{figure}
	
	\begin{figure}[!t]
		\centering{\includegraphics[width=0.8\linewidth]{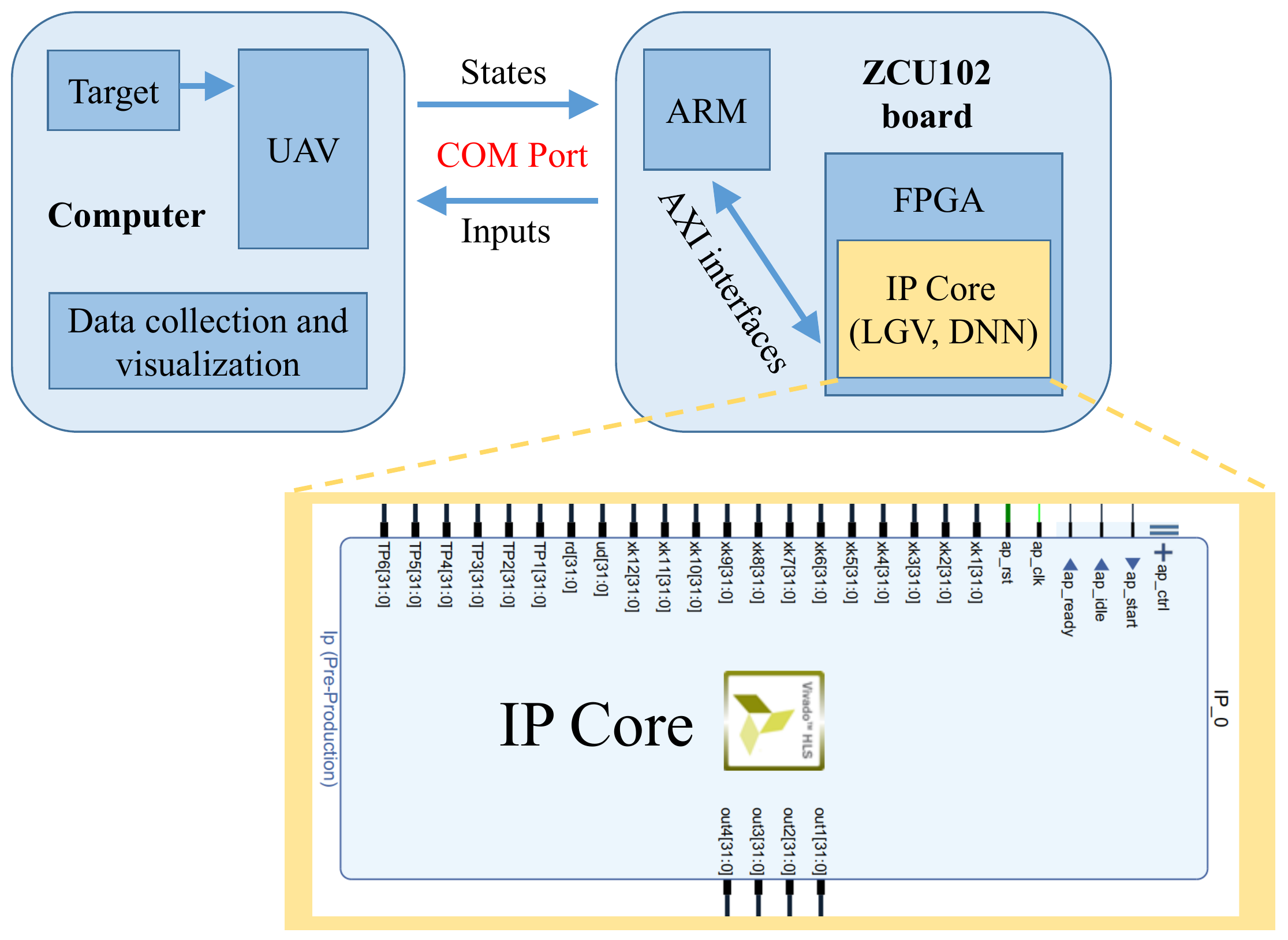}}
		\caption{Schematic diagram of the HIL simulations.}
		\label{figImp}
	\end{figure}
		
	FPGAs can dramatically accelerate the evaluation of DNN by its advantage in parallel processing. We introduce a framework in Fig.~\ref{figWorkflow} to conveniently implement our DNN-based MPCs on an FPGA. Here we adopt the Xilinx ZCU102 as our embedded platform, which contains a Zynq UltraScale+ XCZU9EG-2FFVB1156E MPSoC. The schematic diagram of the HIL simulation is illustrated by Fig.~\ref{figImp}, where the dynamics of the target and UAV are run on the personal computer while the LGV and DNN-based MPC are executed on the FPGA. The computer communicates with the FPGA through the COM port. Besides, an ARM Cortex-A53 processor is adopted to exchange the data between them through AXI interfaces. The digital hardware (HW)  requirements of the implementation are summarized in Table \ref{tab_hw}, which shows that our method requires $9$ percent at most of the HW resources on FPGA. 
		
	Since the DNN-based NMPC and DNN-based NMPC have the same network size in this work, we use the name of DNN-based MPC to represent them for brevity. We also test the evaluation latency on the Intel Core i5 CPU, ARM, and FPGA and record the results in Table~\ref{tab2}. Interestingly, the latency on the FPGA@$200$\si{MHz} is as low as $1.26\times10^{-4}$\si{s}, showing its advantage of parallel computing for the DNN evaluation.  Next, we examine the tracking performances of the DNN-based MPC schemes with three standoff tracking examples.

	\begin{table}[t!]
		\centering
		\caption{A summary of the FPGA implementation}
		\begin{tabular}{lcccc}
			\toprule[1pt]
			Resource  & BRAM  & DSP   & FF    & LUT \\
			\hline
			Available  & 1824  & 2520   & 548160  & 274080 \\
			\hline
			Utilization & 57    & 121    & 27751  & 23167 \\
			\hline
			Utilization (\%) & 3.13  & 4.80   & 5.06  & 8.45 \\
			\bottomrule[1pt]
		\end{tabular}%
		\label{tab_hw}%
	\end{table}%

	\subsubsection{Standoff tracking of a stationary target}\label{sec6e}
	
	\begin{figure}[!t]
		\centering{\includegraphics[width=0.8\linewidth]{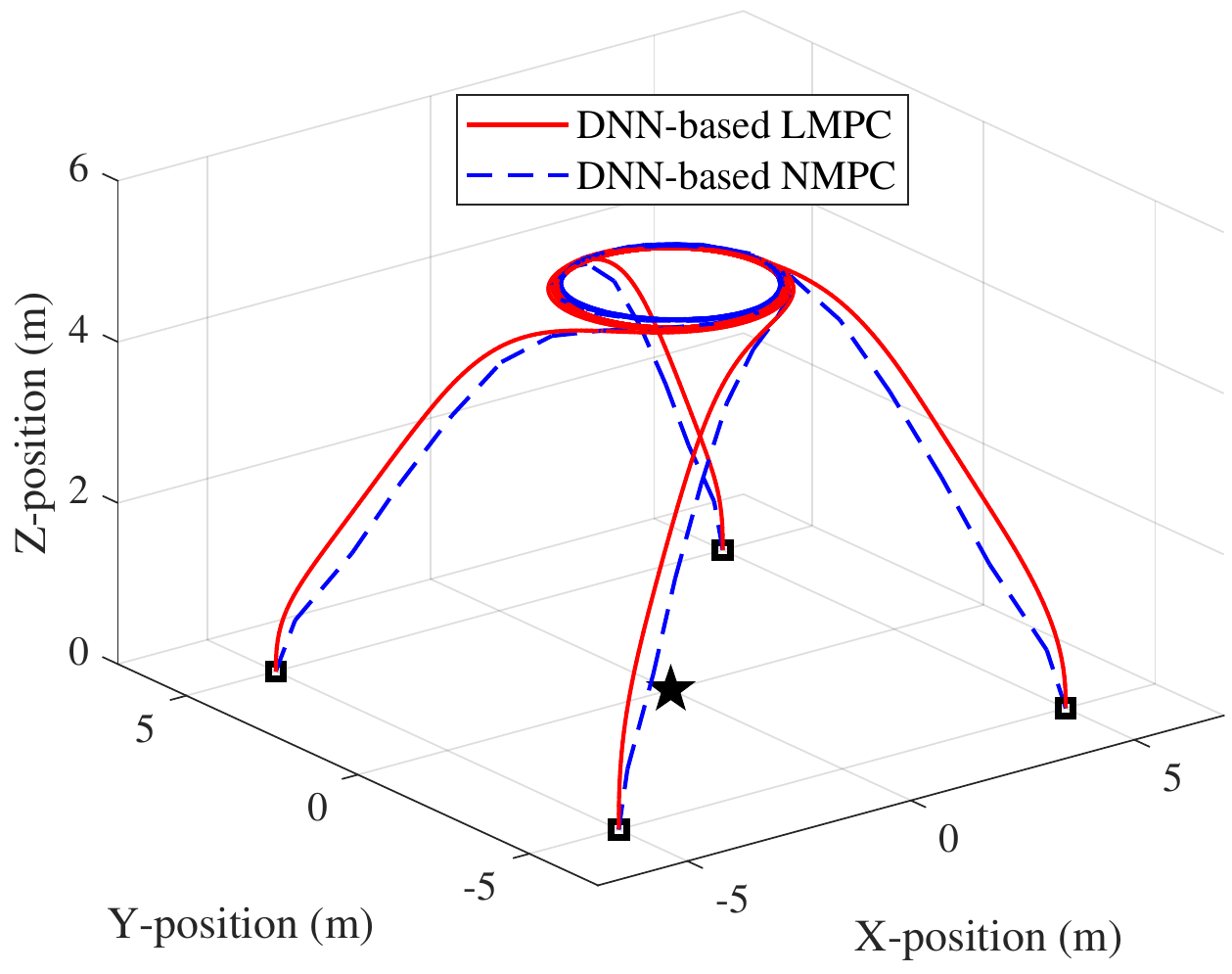}}
		\caption{3D trajectories of the UAV under DNN-based MPCs.}
		\label{fig_tra3}
	\end{figure}
	
	\begin{figure}[!t]
		\centering{\includegraphics[width=0.8\linewidth]{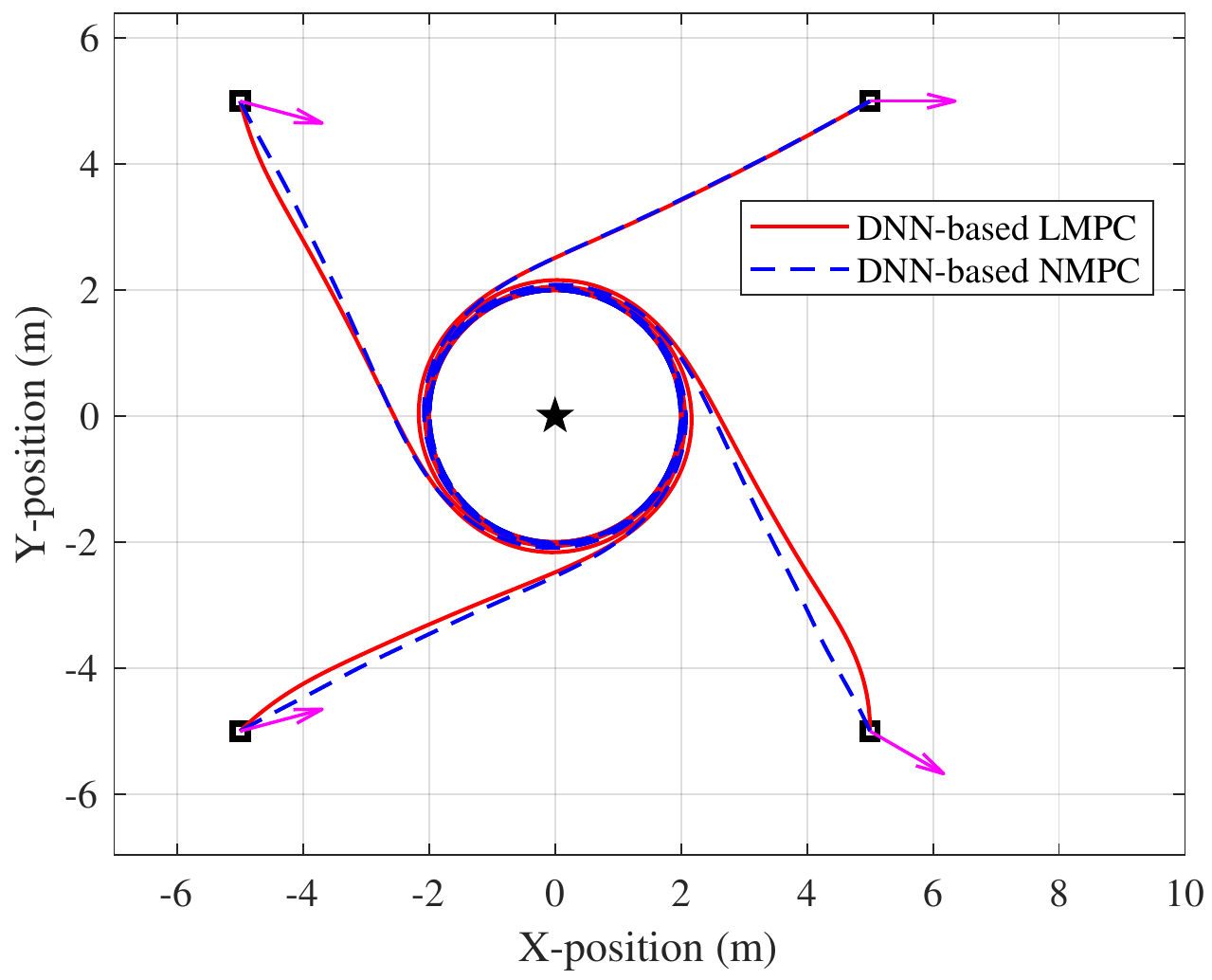}}
		\caption{2D trajectories of the UAV where  squares represent its initial positions $[x_q(t_0),y_q(t_0)]^\T$ and  arrows represent its initial yaw angles $\psi(t_0)$.}
		\label{fig_tra2}
	\end{figure}

	\begin{figure}[!t]
		\centering{\includegraphics[width=0.8\linewidth]{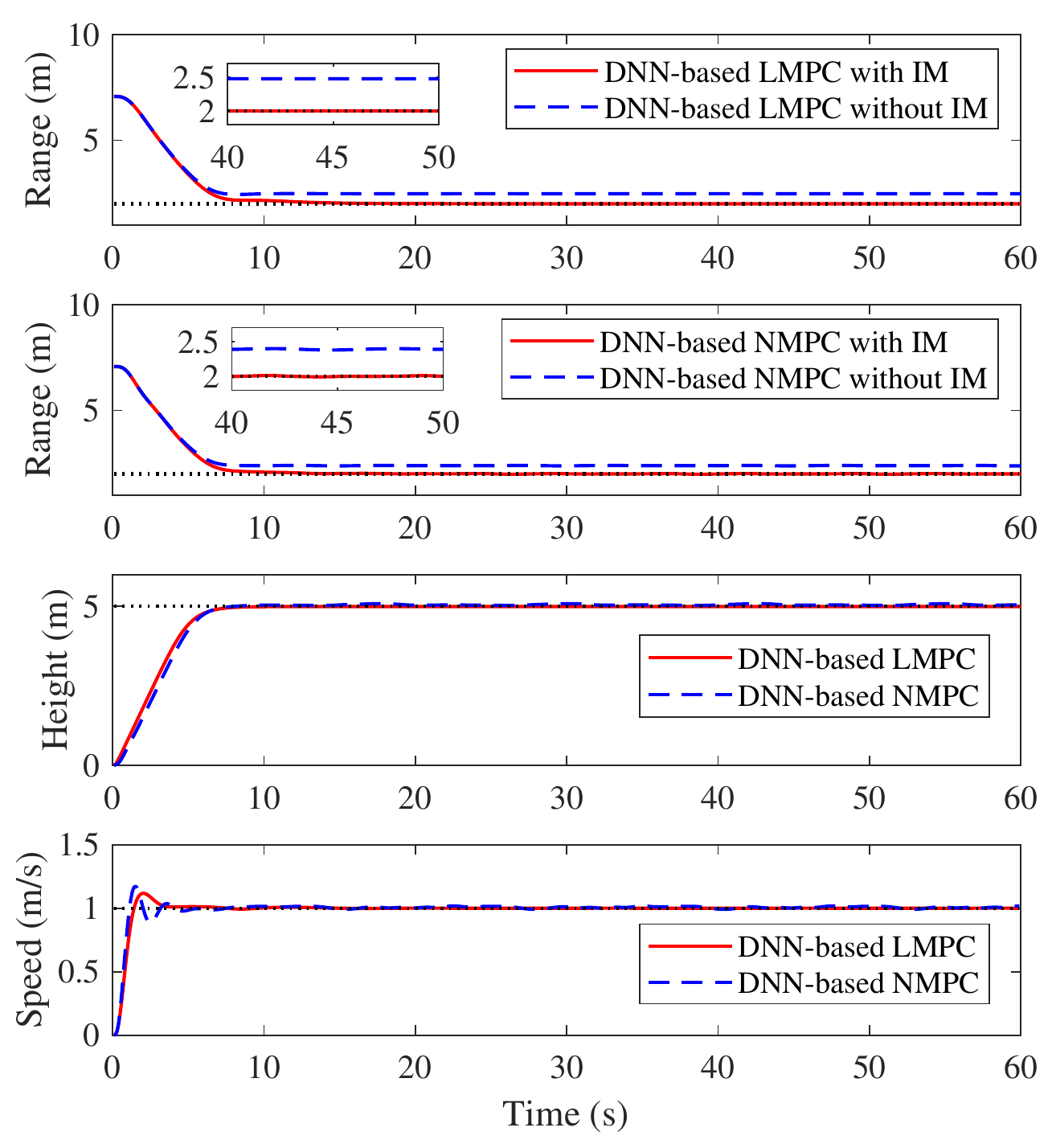}}
		\caption{Relative range, height, and linear speed versus time.}
		\label{figPRS}
	\end{figure}

	\begin{figure}[!t]
		\centering{\includegraphics[width=0.8\linewidth]{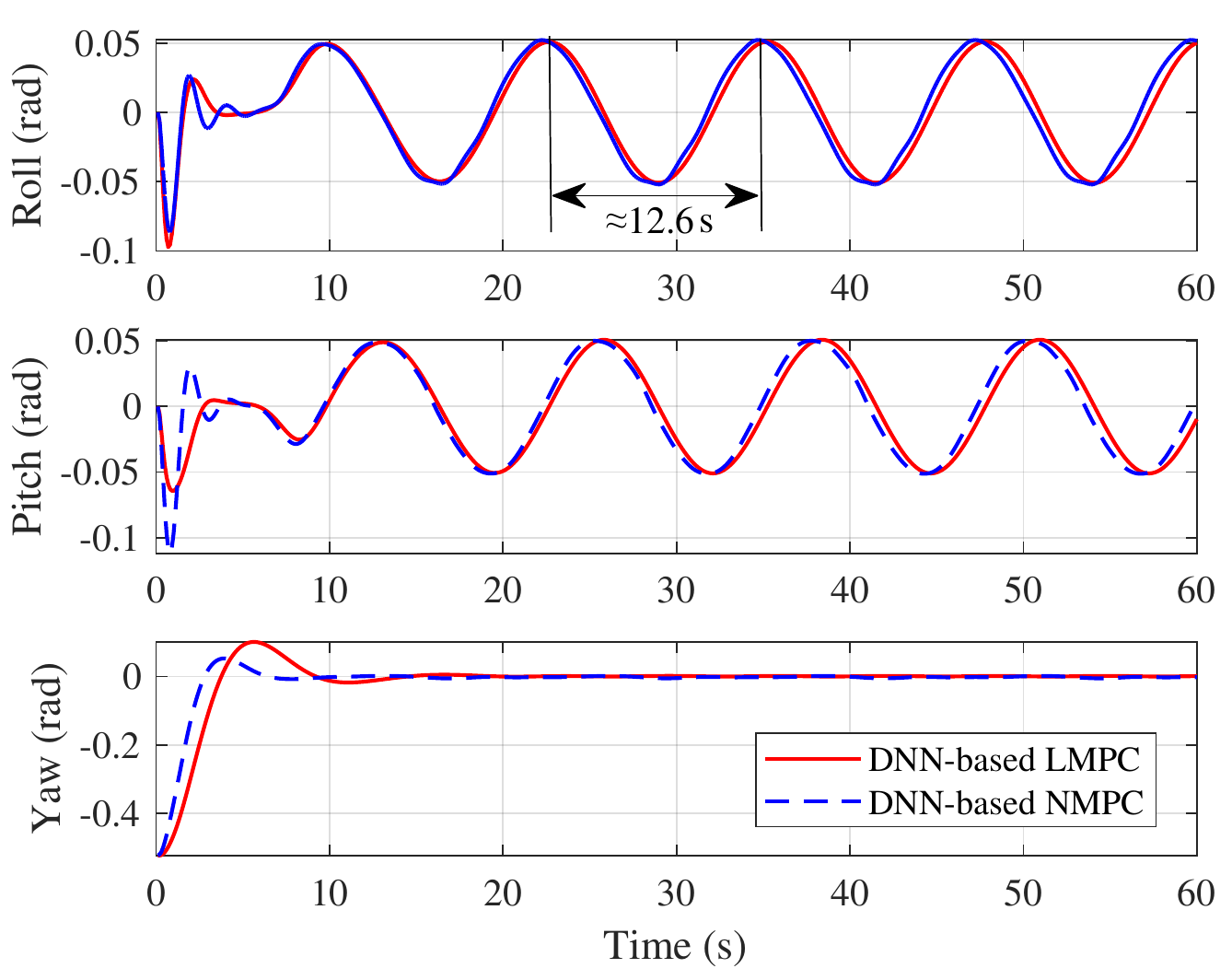}}
		\caption{Euler angles of the UAV versus time.}
		\label{figORI}
	\end{figure}
	
	\begin{figure}[!t]
		\centering{\includegraphics[width=0.8\linewidth]{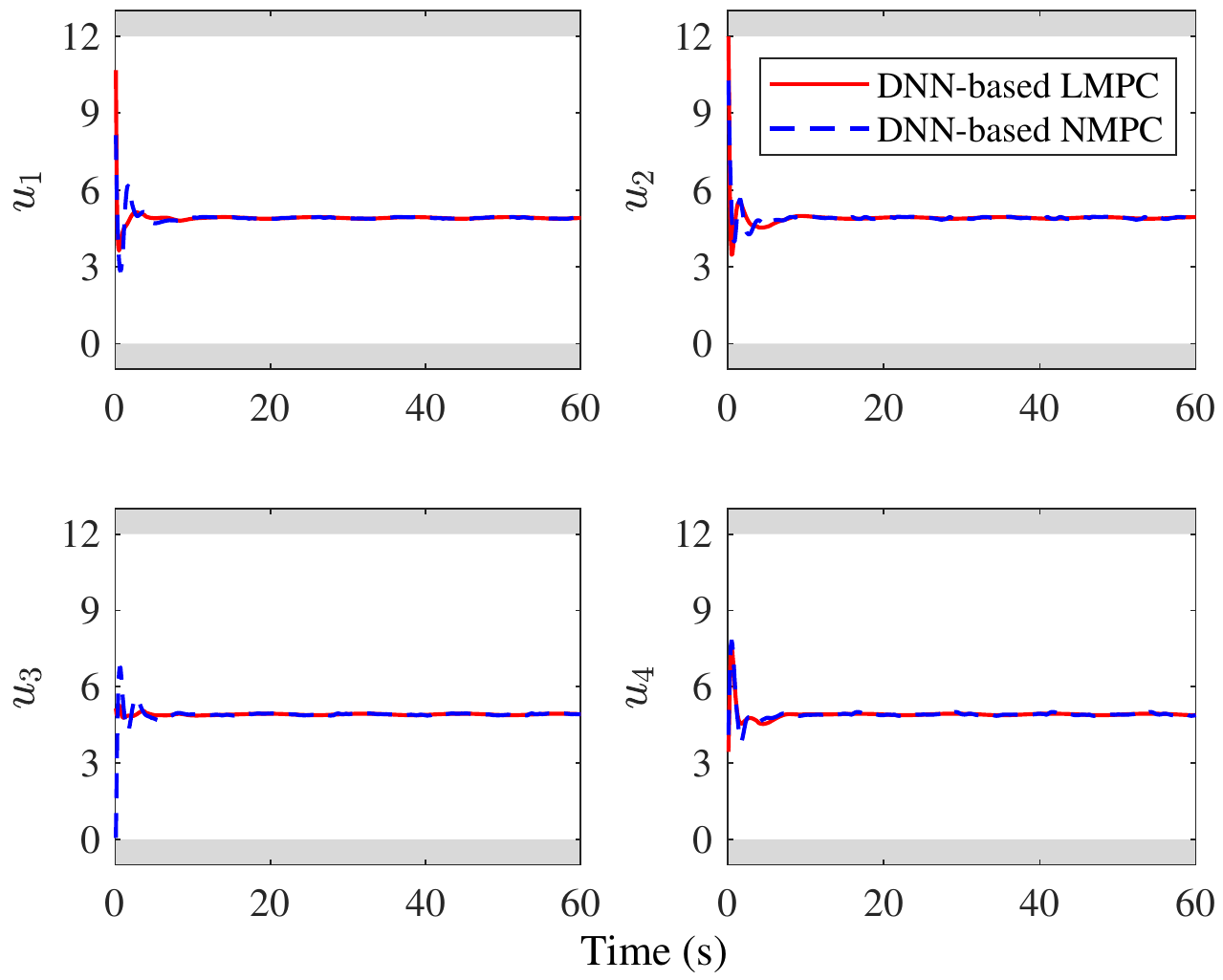}}
		\caption{Control inputs to the UAV versus time.}
		\label{figinput}
	\end{figure}
	
	Fixing the target at the origin, we select four different initial states for the UAV, e.g., $[x_q(t_0),y_q(t_0),\psi(t_0)]^\T=[5,5,0]^\T,[-5,5,-\pi/12]^\T,[-5,-5,\pi/12]^\T,[5,-5,-\pi/6]^\T$, and set $c_1=1.1 \cdot c_2=0.2$ for the IM in \eqref{eqim}.  One can observe from Figs.~\ref{fig_tra3} and \ref{fig_tra2} that all trajectories converge to the desired standoff tracking. Fig.~\ref{figPRS} further verifies that the objective in \eqref{obj} is achieved from the initial state $[5,-5,-\pi/6]^\T$ for both DNN-based MPC schemes
	where the desired values are marked with dotted lines.
	
	By the partially enlarged view of Fig.~\ref{figPRS}, a significant steady-state range error occurs if we remove the IM in \eqref{eqim}, showing its capability in refining the tracking performance.  By Fig.~\ref{figORI}, the yaw angle converges to zero in a short time from $\psi(t_0)=-\pi/6$. However, the overshoot of the linearized one is slightly greater than that of the nonlinear counterpart. Moreover, the UAV needs to persistently adjust its roll and pitch angles at the period of $\approx 12.6$\si{s} to maintain the standoff tracking, which is roughly equal to the theoretical value of $\frac{2\pi r_d}{v_d}=4\pi$. Furthermore, Figs.~\ref{figORI} and \ref{figinput} reveal that the DNN-based MPC is able to satisfy the constraints \eqref{xcons} and \eqref{cons}. Note that the iteration number is select as $3$ for solving \eqref{qp1}, and the projection time via the function $\text{PROJ}(\bm u, j)$ takes up to $6.17\upmu$\si{s} on the FPGA@$200$\si{MHz}. This time has been included in the total computation time of $0.126$\si{ms} in Table \ref{tab2}.
	
	From these results, we can easily observe that  the nonlinear version indeed outperforms  the linearized one in terms of tracking performance,  albeit not  significantly.  Note that the computational cost of generating samples in the nonlinear one is much more expensive.  
	
	\subsubsection{Standoff tracking of a moving target}
	\begin{figure}[!t]
		\centering{\includegraphics[width=0.8\linewidth]{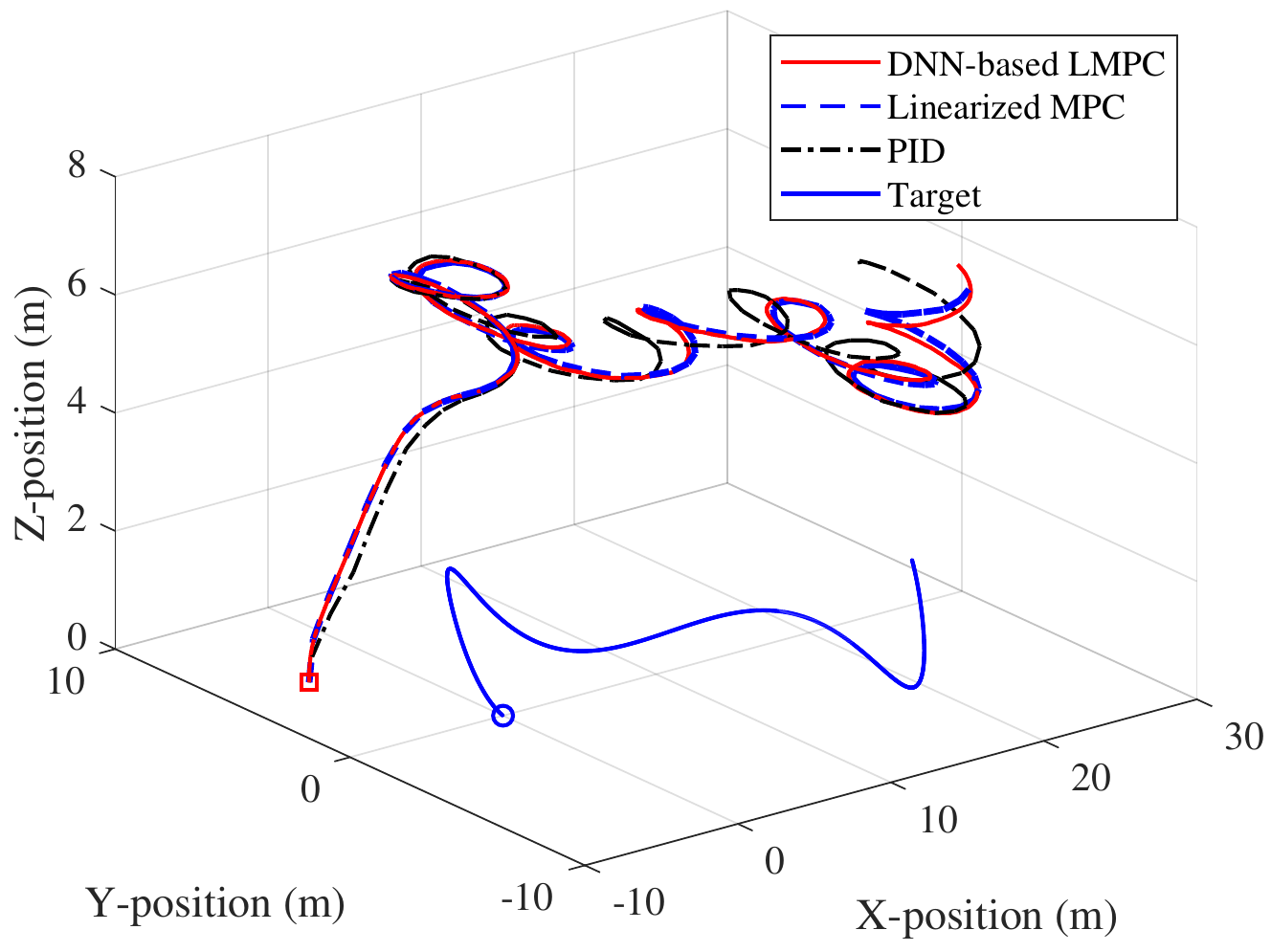}}
		\caption{3D trajectories of the UAV and target.}
		\label{fig3D}
	\end{figure}
	
	\begin{figure}[!t]
		\centering{\includegraphics[width=0.8\linewidth]{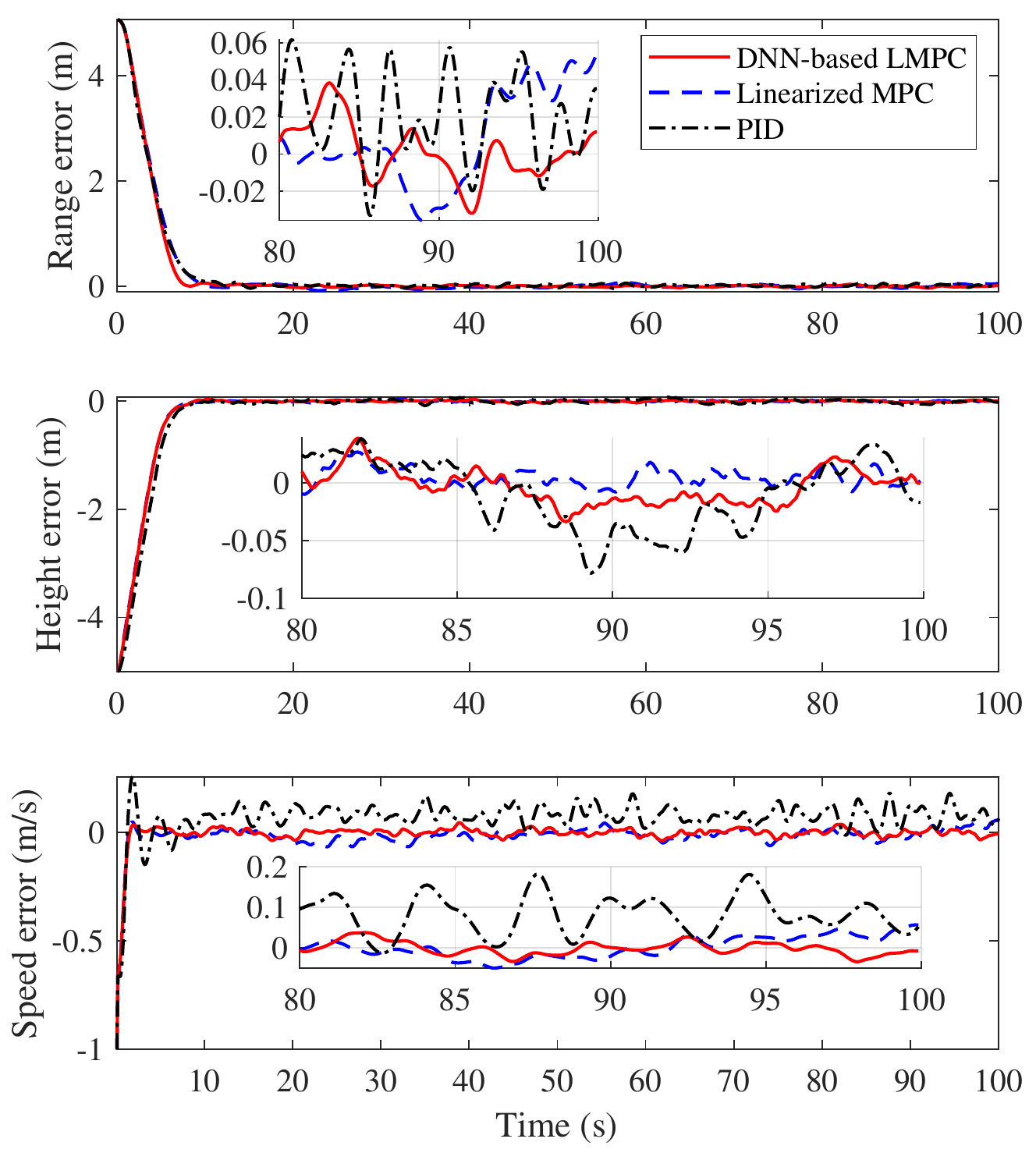}}
		\caption{Tracking errors of the DNN-based LMPC, linearized MPC, and PID.}
		\label{figRHS}
	\end{figure}
	
	Now, we consider the moving target  in Fig.~\ref{fig3D} where its velocity is given as $$\bm v_o(t)=[0.3\sin(0.01t+\pi/6), 0.3\cos(0.06t), 0.1\sin(0.15t)]^\T$$ and add noises to the control input, e.g.,
	\begin{align*}
	\bm u(k) = f_\NN(\bm s(k), \bm \theta_*) + \bm \omega(k),
	\end{align*}
	where $\bm \omega(k)$ is the white Gaussian noise with zero mean and variance of $0.1\cdot I_4$. Moreover, we compare the DNN-based linearized MPC (running on the FPGA) with \eqref{qp} and the PID control \cite{luukkonen2011modelling,gati2013open} (running on the computer). From Fig.~\ref{fig3D}, the tracking results of both the DNN-based MPC and the linearized MPC come close to each other and outperform that of the PID. 
Interestingly,  the DNN-based MPC even has a smaller steady-state range error compared with the linearized MPC thanks to the IM in \eqref{eqim}.  	
\subsubsection{Standoff tracking of a moving target in the presence of obstacles}
	\begin{figure}[!t]
	\centering{\includegraphics[width=0.8\linewidth]{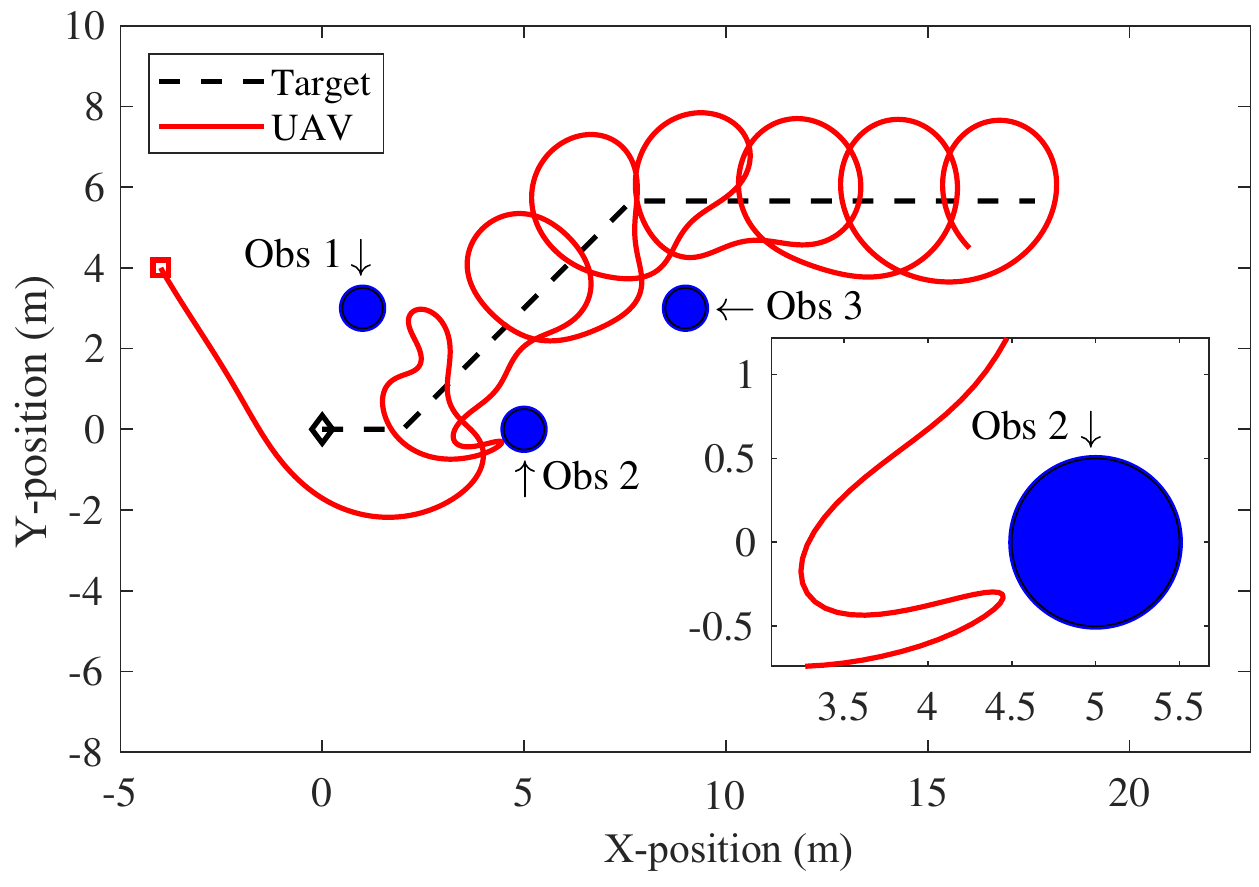}}
	\caption{Trajectory of the UAV with obstacle avoidance.}
	\label{figobs}
	\end{figure}
	
	Consider the standoff tracking in the presence of obstacles whose positions are assumed to be fixed and known. In this example, three circular obstacles with the radius of $0.5$\si{m} are included in Fig.~\ref{figobs}. We integrate the LGV \eqref{lya_cont1} with the inverse convergence vector \cite{wilhelm2019circumnavigation} for guidance, and use the DNN-based MPC for trajectory tracking. In Fig.~\ref{figobs}, the initial positions of the target and UAV are marked by the diamond and square, respectively, and one can observe that the standoff tracking with obstacle avoidance has completed, showing a potential application of our approach in cluttered environments.   
		
	\section{Conclusion} \label{sec7}
	To standoff track a moving target via a UAV, we have proposed an LGV guidance with tunable convergence rates for trajectory planning and a DNN-based MPC with an integral module for trajectory tracking. It was validated by FPGA-in-the-loop simulations that our promising improvements in the computational effort. Specifically, we have demonstrated that the proposed method represents a valid alternative to embedded implementations of MPC schemes, thus allowing its exploitation to more complex systems and applications, not limited to the one presented as test case in this paper. In the future work, we shall conduct flight tests to further validate its real-world performance.

	\section*{Acknowledgement}
	The authors would like to thank the Associate Editor and anonymous reviewers for their very constructive comments, which greatly improve the quality of this work.

	\bibliographystyle{IEEEtran}
	\bibliography{bib/mybib}
	
	\begin{IEEEbiography}
		[{\includegraphics[width=1in,height=1.25in,clip,keepaspectratio]{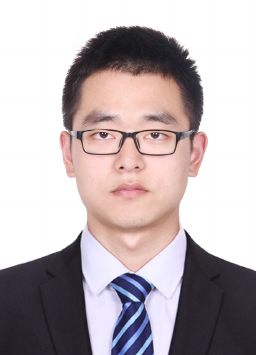}}]
		{Fei Dong} received the B.S. degree from the School of Control Science and Engineering, Shandong University, Jinan, China, in 2014, the M.S. degree from the School of Instrumentation and Optoelectronic Engineering, Beihang University, Beijing, China, in 2017, and the Ph.D. degree from the Department of Automation, Tsinghua University, Beijing, China, in 2022. 
		
		He is currently a Postdoctoral Researcher in the School of Automation Science and Electrical Engineering, Beihang University, Beijing, China. His research interests include model predictive control and learning-based control.
		\end{IEEEbiography}
	
		\begin{IEEEbiography}
		[{\includegraphics[width=1in,height=1.25in,clip,keepaspectratio]{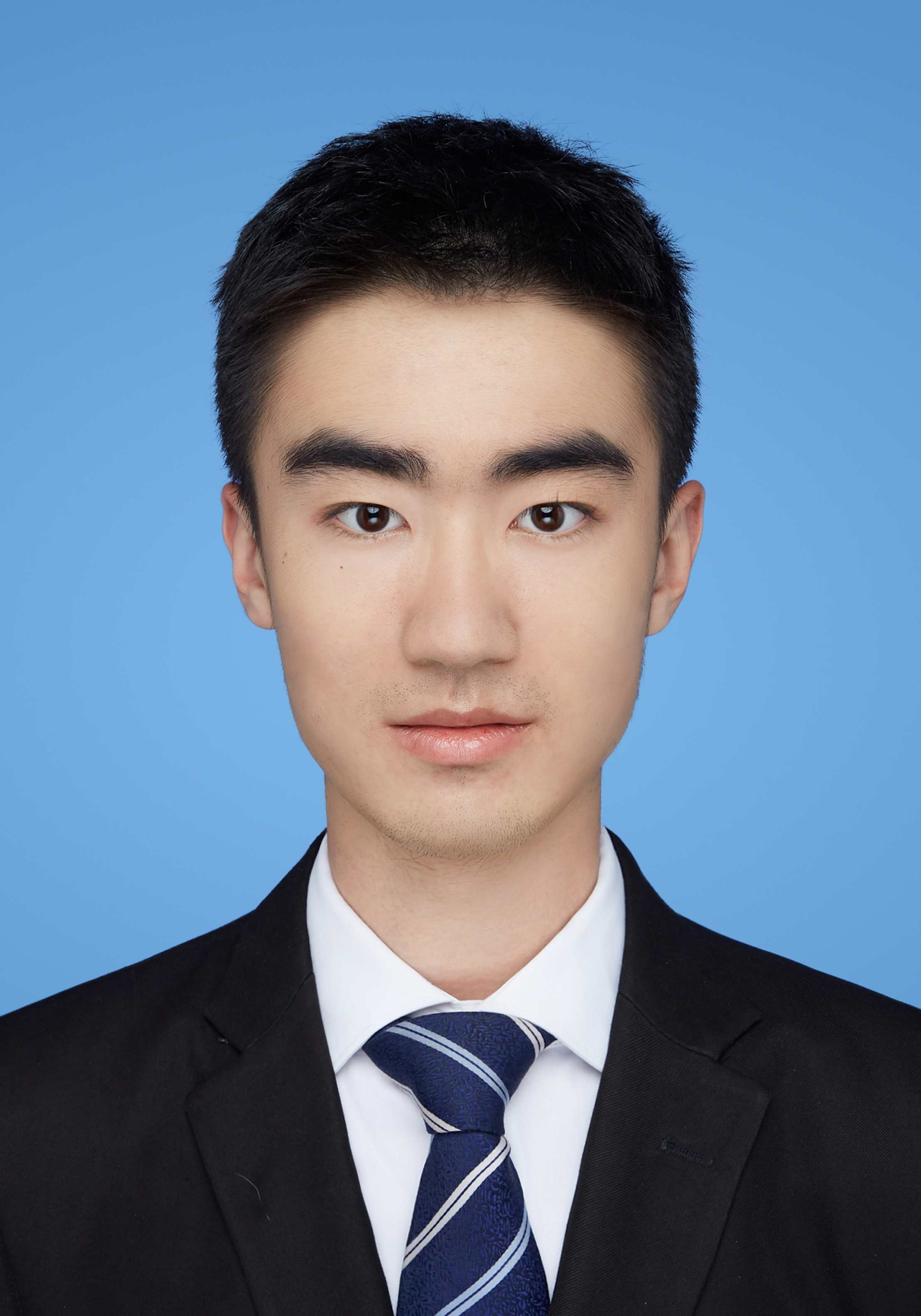}}]
		{Xingchen Li} received the B.S. degree from the Department of Automation, Tsinghua University, Beijing, China, in 2021. He is currently pursuing the Ph.D. degree at the Department of Automation, Tsinghua University, Beijing, China. His research interests include model predictive control and data-driven control. 
	  \end{IEEEbiography}
		
	\begin{IEEEbiography}
			[{\includegraphics[width=1in,height=1.25in,clip,keepaspectratio]{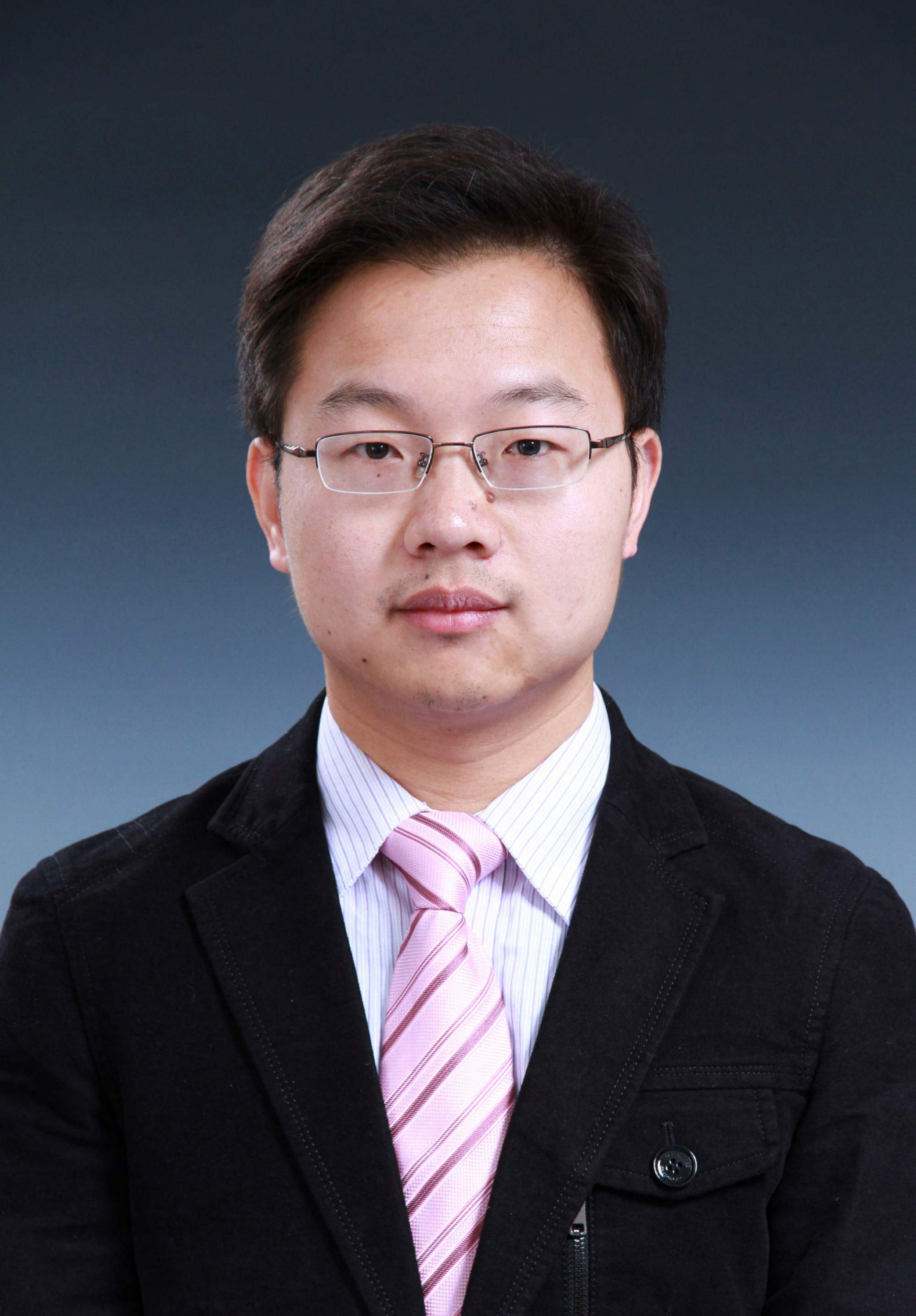}}]
			{Keyou You} (Senior Member, IEEE) received the B.S. degree in Statistical Science from Sun Yat-sen University, Guangzhou, China, in 2007 and the Ph.D. degree in Electrical and Electronic Engineering from Nanyang Technological University (NTU), Singapore, in 2012. After briefly working as a Research Fellow at NTU, he joined Tsinghua University in Beijing, China where he is now a tenured Associate Professor in the Department of Automation. He held visiting positions at Politecnico di Torino,  Hong Kong University of Science and Technology,  University of Melbourne and etc. His current research interests include networked control systems, distributed optimization and learning, and their applications.
	
	Dr. You received the Guan Zhaozhi award at the 29th Chinese Control Conference in 2010,  the ACA (Asian Control Association) Temasek Young Educator Award in 2019 and the first prize of Natural Science Award of the Chinese Association of Automation. He received the National Science Fund for Excellent Young Scholars in 2017, and is serving as an Associate Editor for Automatica,  IEEE Transactions on Control of Network Systems, IEEE Transactions on Cybernetics, and Systems \& Control Letters.
		\end{IEEEbiography}
		
\begin{IEEEbiography}
	[{\includegraphics[width=1in,height=1.25in,clip,keepaspectratio]{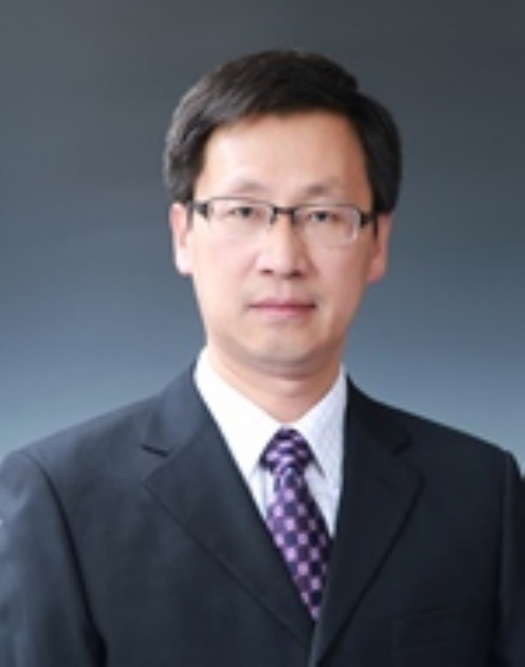}}]
	{Shiji Song} (Senior Member, IEEE) received the	Ph.D. degree in mathematics from the Department of 	Mathematics, Harbin Institute of Technology, Harbin, China, in 1996.
	
	He is currently a Professor with the Department	of Automation, Tsinghua University, Beijing, China.	He has authored more than 180 research papers. His research interests include pattern recognition, system modeling, optimization, and control.
\end{IEEEbiography}

\end{document}